\documentclass[atmp]{ipart_v1}

\Vol{25}
\Issue{2}
\Year{2021}
\firstpage{519}

\usepackage[english]{babel}
\usepackage[utf8]{inputenc}

\usepackage{amsthm}
\usepackage{amsfonts}
\usepackage{mathrsfs}
\newcommand{\lR}{\mathrm{I\hspace{-0.7mm}R}}

\newtheorem{theorem}{Theorem}
\newtheorem{lemma}{Lemma}
\theoremstyle{remark}

\numberwithin{equation}{section}

\begin{document}

%%%%%%%%%%%%%%%%%%%%%%%%%%%%%%%%%%%%%%%%%%%%%%%%%%%%%%%%%%%%%%%%%%%%%%%%%%%%%%%%%%%%%%%%%%%%

%% TITLE %%

%%%%%%%%%%%%%%%%%%%%%%%%%%%%%%%%%%%%%%%%%%%%%%%%%%%%%%%%%%%%%%%%%%%%%%%%%%%%%%%%%%%%%%%%%%%%

\title[Static black holes]{Static black holes in higher dimensional Einstein-Skyrme models}

\author[B. E. Gunara, et al.]{Bobby E. Gunara\footnote{Corresponding author}, Fiki T. Akbar,  Rizqi Fadli,\\ Deden M. Akbar, and Hadi Susanto}

%%%%%%%%%%%%%%%%%%%%%%%%%%%%%%%%%%%%%%%%%%%%%%%%%%%%%%%%%%%%%%%%%%%%%%%%%%%%%%%%%%%%%%%%%%%%

%% Abstract %%

%%%%%%%%%%%%%%%%%%%%%%%%%%%%%%%%%%%%%%%%%%%%%%%%%%%%%%%%%%%%%%%%%%%%%%%%%%%%%%%%%%%%%%%%%%%%

\begin{abstract}

In this paper we construct a class of  hairy static black holes of higher dimensional Einstein-Skyrme theories with the cosmological constant $\Lambda \le 0$ whose scalar is an $SU(2)$ valued field.   The spacetime is set to be conformal to $ \mathcal{M}^4 \times \mathcal{N}^{N-4}$ where $\mathcal{M}^4$ and $\mathcal{N}^{N-4}$ are  a four dimensional spacetime  and a compact  Einstein $(N-4)$-dimensional submanifold for $N \ge 5$, respectively, whereas  $N=4$ is the trivial case. We discuss the behavior of solutions near the boundaries, namely, near the (event) horizon and in the asymptotic region. Then, we establish  local-global existence of black hole solutions and  show that black holes with finite energy exist if their geometries are asymptotically Ricci flat. At the end, we perform a linear stability analysis using perturbative method and give a remark about their stability.

\end{abstract}

\maketitle
\tableofcontents

%%%%%%%%%%%%%%%%%%%%%%%%%%%%%%%%%%%%%%%%%%%%%%%%%%%%%%%%%%%%%%%%%%%%%%%%%%%%%%%%%%%%%%%%%%%%

% % I. Introduction % %

%%%%%%%%%%%%%%%%%%%%%%%%%%%%%%%%%%%%%%%%%%%%%%%%%%%%%%%%%%%%%%%%%%%%%%%%%%%%%%%%%%%%%%%%%%%%

\section{Introduction}
%\section{Introduction and The Main Results}
\label{sec:intro}

%This is introduction  

Over two decades studies of the Einstein-Skyrme model that possesses a black hole solution show that the solution %this black hole 
could be a counterexample to the no-hair conjecture for black holes. Originally, this conjecture comes from stationary spacetime \cite{RuffiniPT1971}, but it could be simplified to the case of static spacetime which states that a static black hole is only characterized by its mass and its electric and magnetic charges. So far, most of the Einstein-Skyrme theories being considered are in four dimensions. In the case of asymptotically flat spacetimes we could also see several models in,  for example,  \cite{Volkov:2016ehx} for an excellent review.  Few examples of  asymptotically anti-de Sitter spacetimes are considered in \cite{ShiikiPRD2005, Shiiki:2005xn, Perapechka:2016cof}, while we have only one example for  de Sitter background \cite{Brihaye:2005an}. In the case of higher  dimensions, only few examples have been considered. In five dimensions, some authors have considered a Skyrme model which can be viewed as an $O(5)$ sigma model coupled to gravity \cite{Brihaye:2017wqa}. The model assures the existence of a topological charge which can be defined globally on the spacetime. Although the results there were found for an $O(5) $ sigma model in five dimensions, the existence of some universality for the basic properties of solutions can also be applied to any dimension higher than five. However, as the dimension increases, there are more terms in the theory  which makes it %the theory 
more difficult to consider. In  seven dimensions, we have so called Skyrme  branes \cite{BlancoPillado:2008cp}.  In the model, the skyrmion field lives on a warped codimensional three submanifold of  a seven-dimensional spacetime endowed with warped metric\footnote{Although the idea in this paper looks similar to that of \cite{BlancoPillado:2008cp}, we will see later that the physics here is completely different. In \cite{BlancoPillado:2008cp}, the authors  introduce seven-dimensional spacetime which is conformal to $\lR^{1,3} \times \mathcal{S}^3$ where $\lR^{1,3}$ is four dimensional Minkowski spacetime and $\mathcal{S}^3$ is three dimensional manifold  conformal to $\lR^+ \times S^2 $  with $\lR^+ $ and $S^2$ being positive real number and  two-sphere, respectively. Then, they examine the physical effect of the additional three dimensions (which is $\lR^3$). In our case,  the extra dimensional submanifold is a compact Riemannian of dimension $N-4$ with $N \ge 4$.}. As a result, the model exhibits a brane world scenario.

In this paper we are trying to establish some results of higher dimensional ($N \ge 4$) Einstein-Skyrme models with the cosmological constant $\Lambda$ turned on. We particularly consider a special class of static solutions where the spacetime is static and chosen to be conformal to $ \mathcal{M}^4 \times \mathcal{N}^{N-4}$ where $\mathcal{M}^4$ and $\mathcal{N}^{N-4}$ are  the four dimensional spacetime  and the compact  $(N-4)$-dimensional submanifold, respectively, with the metric functions $\delta(r), ~ m(r)$, and $C(r)$. In particular, the function $C(r)$ is a warped factor of the compact submanifold $\mathcal{N}^{N-4}$ which has to be Einstein with the cosmological constant $\Lambda_{N-4} $ in order to be compatible with the Einstein field equation for $N \ge 6$, while in the case of $N=4,5$ we have trivial cases, that is, $\Lambda_{N-4} \equiv 0$. Here, we only consider the $\Lambda_{N-4} \ge 0$ case. The Skyrme field is an $SU(2)$ valued field\footnote{In four dimensions this is commonly referred to as a chiral field, which means that it is a specific  nonlinear sigma model where the field values are in a Lie group that can be viewed as an effective QCD-like theory. The chirality corresponds to the fermion representation in four dimensional spacetime. We could particularly extend this property to any even dimension, see for example, \cite{DateLMP1987}.} defined locally on the hypersurface  $\mathcal{S}^3 \subseteq \mathcal{M}^4 $, where $\mathcal{S}^3$ is conformal to $\lR^+ \times S^2 $  with $\lR^+ $ and $S^2$ being positive real number and  two-sphere, respectively, which can be expressed in terms of a profile function $f$ with $f \equiv f(r)$ where $r$ is  the radial coordinate \cite{LuckPLB1986}. 

The setup has some consequences as follows. First, we could not have a global covariant topological current such as baryon number for $N \ge 5$. However, the number could be locally defined  on  $\mathcal{S}^3$.  Second, the dynamic of the theory is indeed more complicated where as we will see, we need to specify the behavior of $\delta(r), ~ m(r)$, $C(r)$, and $f(r)$ on the boundaries,  namely, at the (event)  horizon and the outer boundary,   in order to have a physical black hole with %including its 
  linear stability. Finally, another essential issue is the existence of branches of solutions which can be viewed in this case as two parameter families of black holes. These parameters are the shooting parameters, namely,  the value of $C(r)$ and $f(r)$ on the horizon. Since our model has a complicated structure, one can no longer use these two parameters for the definition of branches. Instead, we use the existence of two-surfaces near the horizon to classify branches.
  
 %\Bnote{We must examine the effects for the normal skyrmions black holes in terms of existence of the compact extradimensions. We must describe in more detail their purpose of the model and its expected outcomes!}

The first consequence follows from the fact that the scalar field $U$ is only an $SU(2)$ valued for  $N \ge 5$. To have a global covariant topological current, the scalar field $U$ has to be  an $O(N)$ valued for  $N \ge 5$ which makes the theory even more complicated because more terms arise as the dimension increases \cite{Brihaye:2017wqa}. To see the second consequence, we have to employ the analysis on the (event) horizon and the outer boundary, namely, the asymptotic region for $\Lambda \le 0$ and near the cosmological horizon for $\Lambda > 0$. Near the (event) horizon, in order to have a consistent picture, all metric functions and the profile function $f$ should be positive constants and they are related to the cosmological constants $\Lambda, ~ \Lambda_{N-4} $  via an inequality coming from  the consistency condition of Ricci scalar such that near the horizon the 4-spacetime becomes ${\mathcal T}^2 \times S^2 $   where the 2-surface ${\mathcal T}^2$ could be either a flat Minkowski surface $\lR^2$ or  an anti-de Sitter surface $AdS_2$ \cite{Kunduri:2007}. A  \textit{primary}  branch  contains black holes with two possible topologies of  ${\mathcal T}^2$, either  $\lR^2$ or   $AdS_2$.  Another branch, that we call \textit{secondary}, %because this branch 
contains black holes with only one possible topology of ${\mathcal T}^2$. In this setup, one can immediately see that this classification  is simply  determined by the value of the cosmological constant $\Lambda$.  On the outer boundary, for $\Lambda \le  0$ case if all the metric functions $\delta(r), ~ m(r)$, and $C(r)$ converge to constants, then the black hole spacetime becomes Einstein geometry. We also find that using the Komar integral \cite{Kastor:2008cqg, Kastor:2009cqg, Gunara:2010iu} the black hole mass is generally constant.

In the case of $\Lambda \le 0$, we also establish local-global existence and uniqueness of this skyrmionic black hole solution. First, we show local existence and uniqueness using Picard's iteration and the contraction mapping properties. Then, using the maximal solution technique we prove global existence. By adding particular decay properties of the functions $\delta(r), ~ m(r)$, $f(r)$, and $C(r)$ in the asymptotic region, it can be shown that the global solution is finite.  Moreover, if we demand that this solution has finite energy, then the decay property of  $f(r)$ has a particular form and the cosmological constant $\Lambda$ must vanish. In other words, finite energy black holes exist only if they are asymptotically Ricci flat.

 Finally, we comment on stability of the black holes by deriving a corresponding linear eigenvalue problem of the governing model. However, due to the nature of our asymptotic approach, %i.e., the existence analysis was carried out near the asymptotic region, 
 we cannot establish a clear answer whether or not they are stable. To show stability, it requires all of the eigenfrequencies to be real, while for instability it needs the presence of at least an eigenfrequency with an imaginary part. %using perturbative method. 
 We can only show that there exist a family of solutions with 'eigenfrequency' $\omega^2 > 0$ for any $\Lambda \le 0$.  Our asymptotic analysis fails for $\Lambda > 0$ case because the physical spacetimes are bounded on $[r_H, r_C)$ where $r_H$ and $r_C$ are radius of the (event) horizon and the cosmological horizon, respectively.

The structure of this paper can be mentioned as follows. In Section \ref{sec:einsteinskyrme} we construct generally a class of static black holes of higher dimensional ($N \ge 4$) $SU(2)$ Einstein-Skyrme theory with the cosmological constant $\Lambda$. We discuss the behavior of the solutions near the boundaries, namely, the (event) horizon and the asymptotic region in Section \ref{sec:SolNearBounds}. In Section \ref{sec:ExisSolFinitE} we establish some results on the local-global existence and the uniqueness of the  solutions. We also consider the finiteness of the energy functional. In Section  \ref{sec:StableSol} we derive a linear eigenvalue problem using perturbative method with the detailed computation in Appendix \ref{sec:Strum-LiouvilleEq}. All of these sections are mainly focussed on the $\Lambda \le 0$ case. Finally, we  discuss the $\Lambda > 0$ case in Section \ref{sec:Lambdapositif}. We finish our paper with conclusion and future works in Section~\ref{sec:conclusion}.

\section{Einstein-Skyrme model in higher dimension}
\label{sec:einsteinskyrme}

In this section we construct static hairy skyrmionic black holes in higher dimension ($N \ge 4$). Our starting point is to consider that $N$ dimensional spacetime $({\mathcal M}^N, g_{\mu\nu})$ is a Lorentzian manifold ${\mathcal M}^N$ equipped with a\linebreak Lorentzian metric $g$ which is conformal to $ \mathcal{M}^4 \times \mathcal{N}^{N-4}$ with $\mathcal{M}^4$ being the four dimensional spacetime and $\mathcal{N}^{N-4}$ being the spatial extra dimensional submanifold assumed to be compact. The local coordinates on ${\mathcal M}^N$ are given by $x^\mu = (t,r,\theta,\varphi,x^i)$ where $x^i$ are local coordinates for $\mathcal{N}^{N-4}$. The Greek indices denote the spacetime index, $\mu,\nu = 0,1,\dots,N-1$ and the Latin indices denote extra dimensions index, $i,j = 4,\dots,N-1$. The ansatz metric in this paper is chosen to be static given by
\begin{align}
\label{AnsatzMetric} 
ds^2 &= - e^{2\delta(r)}B(r) dt^2 + B(r)^{-1}dr^2 \\
\notag &\quad + r^2 \left(d\theta^2 + \sin^2 \theta \;d\varphi^2\right) + r^2 C(r) \hat{g}_{ij}(x^i) \;dx^i dx^j ,
\end{align}  
where $r$ is the radial coordinate and the function $B(r)$ has the particular form 
\begin{equation} 
B(r) = \lambda -\frac{2m(r)}{\left(N-3\right)r^{ N-3 }} - \frac{2\mathrm{\Lambda }}{\left(N-1\right)\left(N-2\right)}r^2   , \label{Bdefinition}
\end{equation} 
with $\lambda > 0$ that depends on the Ricci scalars of 2-sphere $S^2$ and $(N-4)$-submanifold $\mathcal{N}^{N-4}$.  The function $C(r)$ must be smooth and be a positive valued function. In this section we mainly discuss the model with the cosmological constant $\Lambda \le 0$, whereas for $\Lambda > 0$ case we put the discussion  in Section \ref{sec:Lambdapositif}.  For the ansatz \eqref{AnsatzMetric} to describes a black hole solution, we have to take the following assumptions: In the asymptotic region the behavior of the functions $m(r)$ and $\delta(r)$ is given by \cite{ShiikiPRD2005}
\begin{equation}
m(r) \rightarrow M > 0  ,\qquad \delta(r) \rightarrow 0\quad \mathrm{as} \quad r \rightarrow \infty  , \label{asympcon}
\end{equation}
while at the event horizon $r=r_H$ (see the discussion in the next section), we should have $B(r_H)=0$ such that
\begin{equation} 
 m(r_H) = \frac{\left(N-3\right)\lambda}{2}{r_H}^{\left(N-3\right)} - \frac{\Lambda \left(N-3\right)}{\left(N-1\right)\left(N-2\right)}{r_H}^{\left(N-1\right)}\quad \mathrm{and}\quad  \delta \left(r_H\right) \equiv \delta_H  , \label{quantitiesathorizon}
\end{equation} 
with $\delta_H$ a positive constant. The behavior of $C(r)$ will be discussed in the next section. 

The submanifold  $\mathcal{N}^{N-4}$ has to be Einstein, namely
\begin{equation}
\hat{R}_{ij} = \Lambda_{N-4} \ \hat{g}_{ij}  ,
\end{equation}
where $\Lambda_{N-4}$ is the cosmological constant with $\Lambda_{N-4} \ge 0$ and $\hat{R}_{ij}$ is Ricci tensor in $N-4$ dimensions for $N \ge 6$. Thus, the constant $\lambda$ has to be of the form
\begin{equation}
\lambda = \frac{(2 + (N-4)  \Lambda_{N-4} )}{(N-2)(N-3)}  .
\end{equation}
It is important to note that for $N=5$ this extra dimensional submanifold would be a circle, so the Riemann tensor is trivial, i.e. $\hat{R}_{ijkl} \equiv 0$. Therefore, $\hat{R}_{ij} \equiv 0$ and $\Lambda_1 \equiv 0$ in this case. For $N=4$, we have the trivial case.

Now, let us write down the action of higher dimensional Einstein-Skyrme with cosmological constant, namely, 
\begin{align}
\label{ESAction}
S=\int{\sqrt{-g}\ d^Nx}\biggl[&\frac{1}{2}\left(R-2\Lambda\right)+\frac{F_\pi^2}{16}\ g^{\mu\nu}  {\mathrm{Tr}} \left(L_\mu{L_\nu}^\dag\right)\\
&+\frac{1}{32a^2}g^{\mu\rho}g^{\nu\sigma} \mathrm{Tr} \left(\left[L_\rho,L_\sigma\right]\left[L_\mu,L_\nu\right]^\dag\right)\biggr],\notag 
\end{align}
where $g$ is the determinant of the spacetime metric $g_{\mu\nu}$, $U$ denotes an $SU(2)$ chiral field, and  $L_\mu=U^\dag\partial_\mu U$.  Varying \eqref{ESAction} with respect to the spacetime metric $g_{\mu\nu}$, we obtain the Einstein field equation
\begin{equation}
\label{EFE}
R_{\mu \nu } + \Lambda g_{\mu \nu } = T_{\mu \nu } - \frac{g_{\mu \nu } \left(T-N\mathrm{\Lambda }\right)}{\left(N-2\right)}  ,
\end{equation} 
where $R_{\mu\nu}$ is the Ricci tensor of ${\mathcal M}^N$ and $T_{\mu \nu }$ is the skyrmionic energy-momentum tensor given by
\begin{align} 
\label{EnergyMomentumTensor} 
T_{\mu \nu }  &= \frac{F^2_{\pi }}{8}g_{\mu \gamma }g_{\nu \lambda }\ \left(g^{\gamma \rho }g^{\lambda \sigma } \mathrm{Tr} \left(L_{\rho }{L_{\sigma }}^{\dagger }\right)-\frac{1}{2}g^{\gamma \lambda }g^{\alpha \beta } \mathrm{Tr} \left(L_{\alpha }{L_{\beta }}^{\dagger }\right)\right)  \\
\notag &\quad  + \frac{1}{16a^2}g_{\mu \gamma }g_{\nu \lambda }\left(g^{\gamma \rho }g^{\lambda \alpha }g^{\sigma \beta }\ +g^{\gamma \sigma }g^{\lambda \beta }g^{\rho \alpha }-\frac{1}{2}g^{\gamma \lambda }g^{\alpha \rho }g^{\beta \sigma }\ \right)\\
\notag &\qquad\times  \mathrm{Tr}\left(\left[L_{\rho },L_{\sigma }\right]{\left[L_{\alpha },L_{\beta }\right]}^{\dagger }\right) 
\end{align}
with $F_{\pi }$ and $a$ positive constants. The quantity $T$ is the trace of $T_{\mu \nu } $. To make the discussion of \eqref{EFE} clear, we take the ansatz metric \eqref{AnsatzMetric} and then the skyrmion field $U$ has particularly the form \cite{DateLMP1987, LuckPLB1986}
\begin{equation}
\label{UAnsatz}
U = e^{i \vec{\sigma}\cdot\hat{n}f(r)} = \cos f(r) + \mathrm{i}\; \hat{n}\cdot\vec{\sigma} \sin f(r)  ,
\end{equation}
known as the hedgehog ansatz where $\hat{n}$ is a normal vector in $(N-1)$-dimensional space and $\vec{\sigma} \equiv (\sigma_a)$  with $a = 1, \dots, N-1$ generators of a Lie group. The function $f(r)$ is a skyrmion field profile function. In this paper, we particularly take  $\vec{\sigma} = (\sigma_1, \sigma_2, \sigma_3, 0,\dots,0)$ where $\sigma_1, \sigma_2, \sigma_3$ are the Pauli matrices and $\hat{n}$ is a normal vector whose form is given by
\begin{equation}
\hat{n} = (\sin \theta \cos \varphi , \sin \theta \sin \varphi, \cos \theta, 0, \dots,0 )  . 
\end{equation} 
In other words, the skyrmion field $U$ is locally defined on $\mathcal{S}^3 \subseteq \mathcal{M}^4 $, where $\mathcal{S}^3$ is conformal to $\lR^+ \times S^2 $. This setup further implies that we cannot define globally a covariant topological current for $N \ge 5$ on $ \mathcal{M}^4 \times \mathcal{N}^{N-4}$, but we can have a local topological charge called local baryon number on $\mathcal{S}^3$ defined as
 \begin{equation}
Q = \frac{\epsilon^{abc}}{24 \pi^2} \int tr\left( L_a L_b L_c\right) r^2 \sin \theta dr d\theta  d\phi    ,   \label{Topocharge}
\end{equation}
where $a,b,c = 1,2,3$.  It is worth mentioning that to have a global topological current for $N \ge 5$ in our model, the skyrmion field $U$ should be an $O(N)$  valued field which is similar to  the case constructed  in \cite{Brihaye:2017wqa}. However, the general theory with $O(N)$ symmetry is more complicated than the model considered in that reference. We leave this general model for future work.

Inserting the ansatzs \eqref{AnsatzMetric} and \eqref{UAnsatz} into \eqref{EFE}, we then obtain  
\begin{align}
 \label{ReducedEinsteineq} 
&(N-4)\left[\frac{\delta' C'}{2C} -\frac{C''}{2C} - \frac{ {C'}^2}{4C^2} - \frac{C'}{r C} \right] + (N-2) \frac{\delta'}{r}\\
\notag &\quad  = \left(\frac{F_{\pi}^2}{4} + \frac{2 \sin^2 f}{a^2 r^2}  \right) {f'}^2 ,\\%
 \label{ReducedEinsteineq2}
 & - r (B' + B \delta') -(N-3) B+1 - (N-4) \frac{C' B r}{2 C}\\
\notag  &\quad = \frac{F_{\pi}^2}{4} \sin^2f 
  + \frac{(N-4)}{(N-2)a^2} B {f'}^2 \sin^2f \\
\notag &\qquad  +  \frac{(N-3)}{(N-2)} \frac{\sin^4f}{ a^2 r^2} + \frac{2\Lambda r^2 }{(N-2)}  ,
\end{align}
\begin{align}
\label{ReducedEinsteineq3}
&  -\frac{1}{2} C'' B r^2 - (N - 6) \frac{ {C' }^2 B r^2}{4C} - (N-3) B C' r\\
& -\frac{1}{2} B' C' r^2 - B' C r - (N-3)B C   \notag \\ 
& - \frac{1}{2} B \delta' C' r^2 - B \delta' C r  + \Lambda_{N - 4}\notag \\
&\qquad  = - \frac{2 B C}{(N-2)a^2} {f'}^2 \sin^2f -  \frac{C \sin^4 f}{a^2 r^2}  + \frac{2\Lambda C r^2}{(N-2)}  . \notag
\end{align}

 Then, to get the equation of motions of the function $f(r)$, we consider the static energy of this higher dimensional skyrmionic black hole that %can be defined using $T_{00}$ whose form 
 is given by 
\begin{align} 
\label{Estatic} 
E & =  4\pi A({\mathcal N}^{N-4}) \int_{r_H}^{+\infty} dr   r^{N-2}  C^{\frac{(N-4)}{2}} e^{\delta}\\
\nonumber&\quad\times \left[\frac{F^2_{\pi }}{8} \left( B \left(f'\right)^2 + \frac{2\sin^2 f}{r^2} \right) \right.\\
 \nonumber 
&  \qquad \left. + \frac{\sin^2 f}{a^2 r^2} \left(  B (f')^2 +\frac{\mathrm{sin}^2 f\ }{ r^4}\right)\right]  ,
\end{align} 
where $A({\mathcal N}^{N-4})$ is the volume of the compact submanifold $\mathcal{N}^{N-4}$ for $N \ge 4$  and $A({\mathcal N}^0) =1$. By varying \eqref{Estatic} with respect to $f$, we obtain the equation of motions for the  field $f$ to be given by
\begin{align}
\label{SkyrmionEq}
&  B  f^{\prime\prime} + B' f' +  B \delta' f'  + (N-4)  \frac{B C' f'}{2C} \\
& - \Bigg[  \frac{\sin 2f}{2 r^2} \left(\frac{F_\pi^2}{4} + \frac{\sin^2{f}}{a^2 r^2}\right)  \nonumber \\
 &\quad - \left( (N-2)\frac{F_\pi^2}{8} + (N-4)\frac{\sin^2 f }{a^2 r^2} + \frac{f' \sin 2f }{2 a^2 r} \right) \frac{B f'}{r}  \Bigg]\notag\\
&\qquad\times \left(\frac{F_\pi^2}{8} + \frac{\sin^2{f}}{a^2 r^2}\right)^{-1}\notag\\
&\quad = 0  . \notag
\end{align}

It is worth to write down the norm of Riemann tensor related to the ansatz metric  \eqref{AnsatzMetric}

\begin{align}
 \label{normRiem}
 R_{\alpha\mu\beta\nu}R^{\alpha\mu\beta\nu} &=  \left({B}^{\prime\prime}+ 3{B}^\prime\delta^\prime +  B \left( (\delta^\prime)^2 + \delta^{\prime\prime}\right)\right)^2+\frac{2}{r^2}\left({B}^\prime + 2B \delta^\prime\right)^2\\
&\quad  + \frac{2(B')^2}{r^2} +\frac{4(B-1)^2}{r^4}  + \frac{1}{C^2 r^4} \hat{R}_{ijkl} \hat{R}^{ijkl}\notag \\
&\quad + (N-4)\Bigg[ \frac{1}{4 C^2 r^4} \left({B}^\prime + 2B\delta^\prime \right)^2 \left(\left(r^2 C\right)'\right)^2 \notag \\ 
&\qquad + \frac{B^2}{C^2 r^4}\Biggl(2C+4C^\prime r + C^{\prime\prime}r^2  -\frac{\left(\left(r^2 C\right)'\right)^2}{2Cr^2}+\frac{B'}{2B} \left(r^2 C\right)' \Biggr)^2 \notag\\
&\qquad  + \frac{2{B}^2}{C^2r^6} \left(\left(r^2 C\right)'\right)^2 
  + \frac{(N-5) {B}^2}{8 C^4 r^8} \left(\left(r^2 C\right)'\right)^4 \notag\\ 
&\qquad -  \frac{\Lambda_{N-4} B}{ C^3 r^6} \left(\left(r^2 C\right)'\right)^2 \Bigg]  ,
\notag
\end{align}
and Ricci scalar in this case has the form
\begin{align}
 \label{Ricciscalar}
R &= -B''-2 B \delta'' - 3 B' \delta' - 2 B(\delta')^2 -  \frac{2B }{r^2}\\
&\quad   - (N-4) \frac{B }{C r^2} \left( C'' r^2 + 4r C' + 2C\right)  \nonumber \\
&\quad - (N-4) (N-7) \frac{B}{4 C^2 r^4} \left(  \left( r^2 C\right)' \right)^2 \nonumber\\
&\quad  -  \frac{4}{r} \left(  B' +  \delta' B\right)  - \frac{(N-4) }{C r^2}  \left( B' +  \delta' B\right) \left( r^2 C\right)'  \nonumber \\
&\quad - (N-4)  \frac{2 B}{C r^3}  \left( r^2 C\right)' + \frac{2}{r^2} + (N-4)\frac{\Lambda_{N-4}}{r^2} ,\nonumber
\end{align}
which will be useful for our analysis later.

\section{Solutions near boundaries}
\label{sec:SolNearBounds}

 In this section we will discuss the behavior of the functions $\delta(r)$, $m(r)$, $C(r) $, and $f(r)$ in the near horizon limit  ($r \to r_H$) and in the asymptotic limit  ($r \to + \infty$) which are related to the geometries around the regions. In particular, we consider only  the model with $\Lambda \le 0$. \\

\subsection{Near horizon}
\label{subsec:NearHor}

In the near horizon region, the  functions $\delta(r)$, $m(r)$, $C(r) $, and $f(r)$ can be expanded as
\begin{align}
\delta(r) &= \delta_H + \delta_1 (r-r_H) + O\left((r-r_H)^2 \right)  , \nonumber\\
m(r) &=  m(r_H) + m_1 (r-r_H) + O\left((r-r_H)^2 \right)   ,  \nonumber\\
C(r) &= C_H + C_1 (r-r_H) + O\left((r-r_H)^2 \right)   ,  \nonumber\\
f(r) &= f_H + f_1 (r-r_H) + O\left((r-r_H)^2 \right)  , \label{Expandhorizon}
\end{align}
where $\delta_H, m(r_H) , C_H, f_H$  are positive constants, while $\delta_1, m_1 , C_1, f_1$ are real constants. Inserting the expansion \eqref{Expandhorizon} into \eqref{ReducedEinsteineq}, \eqref{ReducedEinsteineq2}, \eqref{ReducedEinsteineq3} and \eqref{SkyrmionEq} evaluated at $r =r _H$, we have some results as follows: 
\begin{equation}
  m_1 = \frac{1}{2} (N-3) r_H^{N-4}  \left(  \frac{F_{\pi}^2 }{4}\sin^2f_H +  \frac{(N-3)}{(N-2)} \frac{\sin^4f_H}{ a^2 r_H^2}  +  (N-3)\lambda -1\right)    , \label{SolEinsteinExpandhorizon}
 \end{equation}
 \begin{equation}
C_1 =  - \frac{ 2 C_H\left(   \frac{\Lambda_{N-4}}{C_H }-1 +  \frac{F_{\pi}^2 }{4 } \sin^2f_H +  \frac{(2N-5)}{(N-2)} \frac{ \sin^4f_H}{a^2 r_H^2}   \right)}{r_H \left(  \frac{F_{\pi}^2 }{4} \sin^2f_H   + \frac{(N-3)}{(N-2)} \frac{ \sin^4f_H}{ a^2 r_H^2} + \frac{2 \Lambda r^2_H}{(N-2)} - 1\right)}     , \label{SolEinsteinExpandhorizon1} 
 \end{equation}
  \begin{equation}
f_1 =  - \frac{\sin 2 f_H \left( \frac{ F_{\pi}^2}{4} +  \frac{\sin^2 f_H}{a^2 r_H^2}  \right)}{2 r_H^3    \left( \frac{ F_{\pi}^2}{8} +  \frac{\sin^2 f_H}{a^2 r_H^2}  \right) \left(  \frac{F_{\pi}^2 \sin^2f_H}{4 r^2_H}    + \frac{(N-3)}{(N-2)} \frac{ \sin^4f_H}{2 a^2 r_H^4}  + \frac{2 \Lambda}{(N-2)} - \frac{1}{r_H^2} \right) }     ,   \label{SolEinsteinExpandhorizon3}
\end{equation}
\begin{equation}
 \delta_1 =   \frac{ \left( \frac{ F_{\pi}^2}{4} +  \frac{2 \sin^2 f_H}{a^2 r_H^2}  \right) f_1^2 +(N-4) \left(  \frac{ C_1 }{2 C_H}   + \frac{2}{r_H}\right)  \frac{ C_1 }{2 C_H}   }{  \left( (N-4) \frac{ C_1 }{2 C_H}   + (N-2) \frac{1}{r_H}\right) }    
  . \label{SolEinsteinExpandhorizon2}
\end{equation}
These show that we have a parameter space   spanned by $f_H,  r_H, C_H, \delta_H, $ 
$N, \Lambda, \Lambda_{N-4}, F_{\pi}, a$. %To make it clearer, let us focus on two examples as follows.
 In the trivial case $N=4$ with $C_H \equiv 0$ and $\Lambda_{N-4} \equiv 0$,  the value of $f_H$ is constrained by the baryon number  \eqref{Topocharge} given by
  \begin{equation}
Q  = - \frac{2}{\pi} \int_{f_H}^{f_0} \sin^2f df = \frac{1}{2\pi} \left( 2(f_H - f_0)+ \sin2f_0 -  \sin2f_H \right)  ,   \label{Topocharge1}
\end{equation}
 where $a,b,c = 1,2,3$ and $f_0 \equiv f(r \to +\infty)$. As discussed in  \cite{ShiikiPRD2005}, taking $f_0 = 0$ if the solution is a black hole, then 
 $f_H <  \pi $ and the charge $Q$ becomes a fractional number. Moreover,  $f_H$ is a shooting parameter and it has been  shown numerically that there exists upper branch solutions in terms of  $f_H \in (0, \pi) $ where $f^{\mathrm{u}}_H$ and $f^{\mathrm{l}}_H$ are upper and lower branch solutions, respectively,  with   $f^{\mathrm{u}}_H > f^{\mathrm{l}}_H$ \cite{ShiikiPRD2005}.  It is worth to notice that if $f_0 \ge 0$, then $f_H \in I_{f_0}$ such that $I_0 = (0, \pi)$ and $Q$ is still a fractional number \footnote{We will see in Subsection \ref{subsec:GlobExis} that  if we want to have finite energy solutions, then $f_0 = 0$.}.

 For $N \ge 5$, we could take $f_H \in I_{f_0} $ since the charge \eqref{Topocharge1} can be locally defined on $\lR^+ \times S^2$, but  the situation is more complicated to be shown numerically  since we have two shooting parameters, namely, $C_H$ and $f_H$ and other seven free parameters. Moreover,  any input of these  parameters   does not give us a complete picture because it is still unclear whether our black holes  are physical and we have a consistent theory.

Therefore, we have to define another definition of branches that is easier to compute in general and gives us the physical picture of a black hole near its horizon.  This can be achieved by analyzing the behavior of Ricci scalar \eqref{Ricciscalar} around the horizon whose form simplifies to 
 \begin{align}
\label{Ricciscalarhor}
R &=  \frac{ F_{\pi}^2 \sin^2f_H }{2 r^2_H} + \frac{(2N^2-13N+14)}{2(N-2)} \frac{\sin^4f_H}{a^2 r_H^4} \\
\notag &\quad  - (N-4) \frac{\Lambda_{N-4} }{C_H r^2_H} + 2\Lambda +  \frac{1}{r^2_H} \left( 2 + (N-4) \Lambda_{N-4}\right)  . 
\end{align}
At this region, the topology of the spacetime changes to ${\mathcal T}^2 \times S^2 \times {\mathcal N}^{N-4}$  such that we have only two consistent pictures on the 2-surface ${\mathcal T} ^2$, namely, either  ${\mathcal T} ^2 \simeq \lR^2$ or  ${\mathcal T} ^2 \simeq AdS_2$ \cite{Kunduri:2007}. So, the four terms in the right hand side in \eqref{Ricciscalarhor} yield a condition
\begin{equation}
  \frac{ F_{\pi}^2 \sin^2f_H }{2 r^2_H} + \frac{(2N^2-13N+14)}{2(N-2)} \frac{\sin^4f_H}{a^2 r_H^4}  - (N-4) \frac{\Lambda_{N-4} }{C_H r^2_H} + 2\Lambda \le 0   , \label{geomhorizoncon}
\end{equation}
where the equality means ${\mathcal T}^2 \simeq \lR^2$.  We define branch as a class of black holes satisfying \eqref{geomhorizoncon}. Using ${\mathcal T} ^2$, we can classify it as follows. We call  a branch to be \textit{primary} if  it consists of black holes with two possible topology of ${\mathcal T} ^2$, namely, either ${\mathcal T} ^2 \simeq \lR^2$ or  ${\mathcal T} ^2 \simeq AdS_2$.  On the other hand, a branch is called \textit{secondary} if it restricts to one possible topology of ${\mathcal T} ^2$. As we will see below, this classification can immediately be  identified by the value of the cosmological constant $\Lambda$.

Let us consider some cases  as follows. In the case of $N=4$ and $\Lambda \le 0$,  since $f_H \in I_{f_0}$ it is easy to see that for any $r_H > 0$ and $\Lambda \le -\frac{a^2 F^4_{\pi}}{48}$ we only have ${\mathcal T}^2 \simeq AdS_2$ which is the secondary branch. For the primary branch, we have either
\begin{equation}
\frac{\sin^2f_H}{r_H^2} \le  \frac{a^2 F^2_{\pi}}{6} \left(1-\left(1+\frac{48 \Lambda}{a^2 F^4_{\pi} }\right)^{1/2}\right)  ,
\label{SyaratN=4kesatu}
\end{equation}
or
\begin{equation}
r^{-2}_H \ge \frac{a^2 F^2_{\pi}}{6} \left(1+\left(1+\frac{48 \Lambda}{a^2 F^4_{\pi} }\right)^{1/2}\right)  ,
\label{SyaratN=4kedua}
\end{equation}
 with $-\frac{a^2 F^4_{\pi}}{48} < \Lambda  \le 0$. Here, it is possible to have either ${\mathcal T}^2 \simeq AdS_2$ or ${\mathcal T}^2 \simeq \lR^2$. In these latter models for  ${\mathcal T}^2 \simeq \lR^2$, we could have some linear stable solutions as shown numerically in \cite{ShiikiPRD2005}.

For $N=5$ and $\Lambda \le 0$, we have for any $r_H > 0$ and $\Lambda \le -\frac{3}{16}a^2 F^4_{\pi}$,  ${\mathcal T}^2 \simeq AdS_2$ which is the secondary branch. For the primary branch, we have either
\begin{equation}
\frac{\sin^2f_H}{r_H^2} \le \frac{3}{2} a^2 F^2_{\pi} \left(1-\left(1+\frac{16 \Lambda}{3a^2 F^4_{\pi} }\right)^{1/2}\right)  ,
\label{SyaratN=5kesatu}
\end{equation}
or
\begin{equation}
r^{-2}_H \ge  \frac{3}{2} a^2 F^2_{\pi} \left(1+\left(1+\frac{16 \Lambda}{3a^2 F^4_{\pi} }\right)^{1/2}\right)  ,
\label{SyaratN=5kedua}
\end{equation}
with $-\frac{3}{16}a^2 F^4_{\pi}< \Lambda  \le 0$. In this case, it is possible to have either ${\mathcal T}^2 \simeq AdS_2$ or ${\mathcal T}^2 \simeq \lR^2$.
In the case of $N \ge 6$ and  $\Lambda \le 0$, we have
\begin{align}
\label{SyaratfisisLambdanegatif}
 \frac{\sin^2f_H}{r_H^2}  &\le -\frac{(N-2) a^2 F^2_{\pi} }{2(2N^2 -13N+14)} \\
\notag &\quad  + \Biggl(\frac{(N-2)^2 a^4 F^4_{\pi}}{4(2N^2 -13N+14)^2}  \\
\notag &\qquad\quad+ \frac{2(N-2) a^2  }{(2N^2 -13N+14)} \left((N-4)\frac{\Lambda_{N-4}}{C_H r_H^2}  - 2 \Lambda \right)  \Biggr)^{1/2}  , 
\end{align}
and it appears to be either ${\mathcal T}^2 \simeq AdS_2$ or ${\mathcal T}^2 \simeq \lR^2$ which shows that only the primary branch exists. Note that the $\Lambda > 0$ case  will be discussed in section \ref{sec:Lambdapositif}.
%%%%%%%
%%%%%%%

\subsection{Around the asymptotic region}
\label{subsec:AsympReg}

In the asymptotic region we recall the assumption \eqref{asympcon}   as such that the decrease rate of the functions $\delta(r)$ and $m(r)$ should be greater than $O\left(r^{-(N-3)} \right)$. Therefore, the functions $\delta(r)$ and $m(r)$ should have the form
\begin{align}
\delta(r) &=    O\left(r^{-2(N-3)} \right)  , \nonumber\\
m(r) &=  M + O\left(r^{-2(N-3)} \right)   ,   \label{Expandasymp}
\end{align}
where $M > 0$. In addition, the function $C(r)$ could be expanded as 
\begin{equation}
C(r) =  C_0 +  \frac{\tilde{C}_1}{r^{n}} + O\left(r^{-(n+1)} \right)      , \label{CrExpandasymp}
\end{equation}
with $C_0 > 0, \tilde{C}_1 \in \lR$, and $n \ge 1$. Then, the Ricci scalar \eqref{Ricciscalar} is simplified to
\begin{equation}
R = \frac{2 N \Lambda}{(N-2)}  + O\left(r^{-n} \right)    ,   \label{Ricciscalarasym}
\end{equation}
showing that  the geometry converges to Einstein with cosmological constant $2  \Lambda / (N-2)$ (or Ricci-flat with $\Lambda = 0$). Moreover, it can be shown using \eqref{normRiem} that the spacetime becomes  maximal as $r \gg M$. In this case the skyrmionic scalar $f(r)$ would have the form  
\begin{equation}
f(r) = f_0   +  \frac{\tilde{f}_1}{r^{n_1}}+ O\left(r^{-(n_1 + 1)} \right)  , \label{fasympcon}
\end{equation} 
where $n_1 \ge 1$ and $f_0, \tilde{f}_1 \in \lR$ which shows that it has to be frozen as $r \to +\infty$. As we will see in the next section, the value of $f_0$ and the bound of $n_1$ are related to the finiteness of the energy functional \eqref{Estatic}.

In order to obtain the mass of the black hole, we have to consider the Komar integral \cite{Kastor:2008cqg, Kastor:2009cqg}  
\begin{equation}
{\mathcal{K}} = \int_{\partial\Sigma} dS_{\mu\nu} \left( \nabla^{\mu} \xi^{\nu} + \omega^{\mu\nu} \right) \  , \label{komarint}
\end{equation} 
where $\partial\Sigma$ is the boundary of a spatial hypersurface of the spacetime ${\mathcal M}^N$. The quantity $\xi^{\mu}$ is a Killing vector of ${\mathcal M}^N$, while $\omega^{\mu\nu}$ is an antisymmetric tensor satisfying \cite{Gunara:2010iu}
\begin{equation}
\nabla_{\mu}  \omega^{\mu\nu} = R^{\nu}_{~ \mu} ~ \xi^{\mu} \  . \label{divergenceomega}
\end{equation}
In our case,  $\xi^{\mu}$ is the time-like Killing vector whose form is $\xi^{\mu} = (1, 0,\dots,0)$. Then, the non-zero solution of \eqref{divergenceomega} is
\begin{equation}
\omega^{10} = - \frac{1}{2}(B' + 2 \delta' B) + \frac{M}{4 \pi A(\mathcal{N}^{N-4})} e^{-\delta} r^{-(N-2)} C^{-(N-4)/2} \  , \label{solusiomega}
\end{equation}
while
\begin{equation}
 \nabla^1 \xi^0 =  \frac{1}{2}(B' + 2 \delta' B)   , 
\end{equation}
where $A(\mathcal{N}^{N-4})$ is the volume of the compact submanifold $A(\mathcal{N}^{N-4})$ for $N \ge 6$, $A(\mathcal{N}^1) = 2 \pi$, and  $A(\mathcal{N}^0) = 1$. Thus, the mass of our black hole is given by
\begin{equation}
 M_{\text{BH}} =\lim_{r \to +\infty} M = M  , \label{BHmass}
\end{equation}
where we have used 
\begin{equation}
dS_{10} = \frac{1}{2}  e^{\delta} r^{N-2} C^{(N-4)/2}  dA(S^2)   dA(\mathcal{N}^{N-4})  . 
\end{equation}
Some comments are in order. First, in general the spacetime $\mathcal{M}^N$ cannot converge to Einstein as $r \to +\infty$ unless the function $C(r)$ has the form of \eqref{CrExpandasymp} in the region. If this not the case, then we still have  $C(r)$ and $f(r)$, i.e.\ the case which is excluded in this paper. Second, the second term in the right hand side of \eqref{solusiomega} comes from the non-trivial cohomology of $S^2 \times \mathcal{N}^{N-4}$ which is related to the distribution of the black hole mass. As we have shown above, this mass still converges to a constant for a general form of $C(r)$.

%%%%%%%%%%%%%%%%%%%%%%%%%%%%%%%%%%%%%
%%%%%%%%%%%%%%%%%%%%%%%%

\section{Existence of solutions with finite energy}
\label{sec:ExisSolFinitE}

In this section we show local-global existence and uniqueness of black hole solutions of \eqref{ReducedEinsteineq}, \eqref{ReducedEinsteineq2}, \eqref{ReducedEinsteineq3}, and  \eqref{SkyrmionEq}. We use Picard's iteration and the contraction mapping properties to prove the local existence and the uniqueness. Then, we establish the global existence using maximal solution technique and show using  the expansions \eqref{Expandasymp}, \eqref{CrExpandasymp}, and \eqref{fasympcon} that this solution is finite. Finally, solutions with finite energy functional \eqref{Estatic} are discussed. Here, we still assume $\Lambda \le 0$.

\subsection{Local existence and smoothness}
\label{subsec:LocExisSmooth}

Let us first recall Eqs.\ \eqref{ReducedEinsteineq}--\eqref{ReducedEinsteineq3} and \eqref{SkyrmionEq} for $N \ge 5$. In the case at hand, we have two constraints in this case, namely  \eqref{ReducedEinsteineq} and \eqref{ReducedEinsteineq2}, which can be cast into
\begin{align}
 \label{ReducedEinsteineqconstraints}
  & \delta ' =  \frac{\left( \frac{F_{\pi}^2}{4}  +   \frac{(N-1)}{(N-2)} \frac{2 \sin^2f }{a^2 r^2} \right)  }{\left((N-4) \frac{C'}{C}  + (N-3)  \frac{2}{r}  \right) }  {f'}^2  \\
&\qquad + (N-4) \resizebox{.8\textwidth}{!}{$\frac{ \left( \frac{1}{Br^2}  \left(  \frac{ \sin^4f }{a^2 r^2} - \frac{2 \Lambda r^2}{(N-2)} \right) -     \frac{B' C'}{2BC} +  \frac{B'}{Br}   
 - (N-3)   \frac{1}{r^2}  - (N-7)   \frac{{C'}^2}{4 C^2}  -  (N-4)   \frac{B'C'}{rBC}   \right)  }{\left((N-4) \frac{C'}{C}  + (N-3)  \frac{2}{r}  \right) }$} , \nonumber\\
 & F_1 B' + F_2 B = F_3   ,\nonumber
\end{align}
respectively, where 
\begin{align}
F_1 &\equiv  r -  \frac{(N-4)   \left( 1 + \frac{rC'}{2C} \right)  }{ (N-4)  \frac{C'}{C} + (N-3) \frac{2}{r} }    ,  \\
F_2 &\equiv  (N-3) + (N-4 ) \frac{rC'}{2C} +  \frac{(N-4)}{(N-2)} \frac{  {f'}^2 }{a^2}   \sin^2f   \nonumber\\  
 &\quad  +  (N-4 )  \resizebox{.8\textwidth}{!}{$\left(  \frac{  \left( \frac{F_{\pi}^2}{4}  +   \frac{(N-1)}{(N-2)} \frac{2 \sin^2f }{a^2 r^2} \right)    {f'}^2 - \frac{(N-3)}{r^2} - ( N-7)    \frac{{C'}^2}{4 C^2}  -  (N-4)   \frac{C'}{rC}  } { (N-4)  \frac{C'}{C} + (N-3) \frac{2}{r} } \right)$}  , \nonumber\\
F_3 &\equiv  1-  \frac{F_{\pi}^2}{4} \sin^2f  -     \frac{(N-3)}{(N-2)}  \frac{ \sin^4f }{a^2 r^2} - \frac{2 \Lambda r^2}{(N-2)} \nonumber \\
&\quad - \frac{1}{r}  \left(  \frac{ \frac{ \sin^4f }{a^2 r^2} - \frac{2 \Lambda r^2}{(N-2)}   } { (N-4)  \frac{C'}{C} + (N-3) \frac{2}{r} } \right)  ,  \nonumber
\end{align}
and $B$ is given by \eqref{Bdefinition} and $B'$ has the form
\begin{equation}
B' = - \frac{2 m'}{(N-3) r^{N-3}}  +  \frac{2 m}{ r^{N-2}} - \frac{4 \Lambda r}{(N-1) (N-2) }  . 
\end{equation}
Using \eqref{Bdefinition}, the second equation in \eqref{ReducedEinsteineqconstraints} gives us
\begin{align}
 \label{SolReducedEinsteineqconstrm}
m &=    \frac{1}{2} (N-3) r^{N-3} \left(  \lambda - \frac{2\mathrm{\Lambda }}{\left(N-1\right)\left(N-2\right)}r^2 \right. \\
\notag&\quad \left. -  \exp \left(- \int \frac{F_2}{F_1} dr \right) \left[  \int \frac{F_3}{F_1}  \exp \left(\int \frac{F_2}{F_1} dr \right)  dr + const. \right]  \right)  ,
\end{align}
showing that $ m= m(C, f, C', f', r) $ and $\delta' = \delta'(C, f, C', f', r)  $. It implies that the dynamical variables are   $(C,f)$.  Let $I \equiv [r, r+ \epsilon]$ where $\epsilon$ is a small positive constant and $U\subset \lR^4$ be an open set. We define  the conjugate  fields
\begin{align}
p_C &\equiv C'   , \nonumber\\
p_f &\equiv f'   ,  
\end{align}
such that we can set ${\bf{u}} \equiv ( C, f, p_C, p_f)$.
\begin{lemma}
	\label{lemmalocalLipshitzconstraint}
Both functions $\delta' ({\bf{u}},r )$ and $m({\bf{u}},r )$ in   \eqref{ReducedEinsteineqconstraints}  and \eqref{SolReducedEinsteineqconstrm}, respectively, are locally Lipschitz with respect to $\bf{u}$.	
\end{lemma}
\begin{proof}
 First, we have
 \begin{equation}
\left|  m \right|_U  \le   \frac{1}{2} (N-3) r^{N-3}  \left(  \lambda + \frac{2 |\Lambda | }{\left(N-1\right)\left(N-2\right)}r^2  +   |const._0| \left| \frac{F_3}{F_1} \right|   + |const.|  \right)  , \label{boundedm}
\end{equation}
which is bounded since the functions $C(r)$ and $f(r)$ are at least a $C^2$-real function.  Using similar argument, we can also show that  $\delta' $ is bounded.  Then, for $\bf{u}, \tilde{\bf{u}} \in U$
\begin{align}
&\left|  \delta' ({\bf u}, r) -  \delta' ( \tilde{{\bf u}}, r) \right|_U  \le   \Bigg| \frac{\left( \frac{F_{\pi}^2}{4}  +   \frac{(N-1)}{(N-2)} \frac{2 \sin^2f }{a^2 r^2} \right)  }{\left((N-4) \frac{p_C}{C}  + (N-3)  \frac{2}{r}  \right) }  p_f^2  \nonumber\\
&\qquad + (N-4) \frac{ \left( \frac{1}{Br^2}  \left(  \frac{ \sin^4f }{a^2 r^2} - \frac{2 \Lambda r^2}{(N-2)} \right) -     \frac{B' p_C}{2BC} +  \frac{B'}{Br}   
   \right)  }{\left((N-4) \frac{p_C}{C}  + (N-3)  \frac{2}{r}  \right) }   \nonumber\\
&\qquad - (N-4) \frac{ \left(  (N-3)   \frac{1}{r^2}  + (N-7)   \frac{p_C^2}{4 C^2}  + (N-4)   \frac{B' p_C}{rBC}   \right)  }{\left((N-4) \frac{p_C}{C}  + (N-3)  \frac{2}{r}  \right) } \nonumber\\ 
&\qquad -  \left(C \to \tilde{C},  f \to \tilde{f}, p_C \to \tilde{p}_C,  p_f \to \tilde{p}_f \right) \Bigg|  ,  \nonumber\\
&\left| m ({\bf u}, r) -  m( \tilde{{\bf u}}, r) \right|_U  \le  \frac{1}{2} (N-3) r^{N-3}    |const._0| \left| \frac{F_3}{F_1} ({\bf u}, r)  - \frac{F_3}{F_1}( \tilde{{\bf u}}, r) \right| ,\hspace{-1em} 
\end{align}
 and for any smooth function $F(f)$ we have locally 
\begin{equation}
F( f) - F(\tilde{f} ) \leq \sup_{s\in[0,1]}\left[ F'( f + s(\tilde{f} - f)) \right]  (f -  \tilde{f})  , \label{anyFlocal}
\end{equation}
on $U$, we obtain that $\delta' $ and $m$ satisfy the local Lipschitz condition
\begin{align}
\left|  \delta' ({\bf u}, r) -  \delta' ( \tilde{{\bf u}}, r) \right|_U  &\le C_{ \delta'}(|\bf{u}|, |\tilde{\bf{u}}|) | \bf{u} - \tilde{\bf{u}}|  ,  \nonumber\\
\left| m ({\bf u}, r) -  m( \tilde{{\bf u}}, r) \right|_U  &\le C_m (|\bf{u}|, |\tilde{\bf{u}}|) | \bf{u} - \tilde{\bf{u}}|   ,  \label{localLipshitzconstraints}
\end{align}
on an open set $U\subset \lR^4$ where $C_{ \delta'}(|\bf{u}|, |\tilde{\bf{u}}|)$ and $C_m(|\bf{u}|, |\tilde{\bf{u}}|)$  are bounded positive-valued functions . 
\end{proof}
Next, we  cast \eqref{ReducedEinsteineq3} and \eqref{SkyrmionEq} into
\begin{equation}
\frac{d \bf{u}}{dr} =  \mathcal{J}({\bf u}, r)  ,  \label{ReducedEinsteineqSkyrmionEq1}
\end{equation}
where  
\begin{equation}
  \mathcal{J}({\bf{u}}, r)   \equiv  \left( \begin{array}{c}
p_C  \\
p_f \\
J_C \\
J_f
\end{array} \right)  , \label{fungsiJ}
\end{equation}
where
\begin{align}
J_C &\equiv - (N - 6) \frac{ p_C^2 }{2C} - (N-3) \frac{2p_C}{r} -  \frac{B' p_C}{Br}  -   \frac{B' C}{Br}\\ 
&\quad  - (N-3) \frac{2C}{r^2}  -  \delta' p_C - \frac{C}{r} \delta'   + \frac{2}{B r^2}\Lambda_{N - 4} \nonumber\\
&\quad + \frac{4  C}{(N-2)a^2 r^2} p_f^2 \sin^2f +  \frac{2 C }{a^2 Br^4}  \sin^4 f - \frac{4\Lambda C}{(N-2) B}   , \nonumber\\
J_f &\equiv - \delta'  p_f -  \frac{B' p_f}{B}   - (N-4)  \frac{ C' p_f}{2C}  \nonumber\\
&\quad  +  \frac{\sin{2f}}{2 B r^2} \left(\frac{F_\pi^2}{4} + \frac{\sin^2{f}}{a^2 r^2}\right) \left(\frac{F_\pi^2}{8} + \frac{\sin^2{f}}{a^2 r^2}\right)^{-1}   \nonumber\\
\nonumber &\quad -  \frac{p_f}{r} \left( (N-2) \frac{F_\pi^2}{8} + (N-4) \frac{\sin^2 f }{a^2 r^2} + \frac{ p_f \sin 2f }{2a^2 r^2} \right) \left(\frac{F_\pi^2}{8} + \frac{\sin^2{f}}{a^2 r^2}\right)^{-1}    ,
   \end{align}
with the constraint \eqref{ReducedEinsteineqconstraints} which can be viewed as $\delta' = \delta' ({\bf{u}}, r) $ and $m = m({\bf{u}}, r) $. 

We establish the local existence and the uniqueness of \eqref{ReducedEinsteineqSkyrmionEq1} using the contraction mapping theorem by showing that the operator $\mathcal{J}$ is locally Lipshitz.
\begin{lemma}
	\label{localLipshitz}
Let $\mathcal{J}$ be an operator given in \eqref{fungsiJ} and we have \eqref{localLipshitzconstraints} . Then,  it is locally Lipschitz with respect to $\bf{u}$.	
\end{lemma}
\begin{proof}
From \eqref{fungsiJ}, we obtain the following estimate
\begin{align}
\left|  J_C \right|_U  &\le  \biggl|  (N - 6) \frac{ p_C^2 }{2C} + (N-3) \frac{2p_C}{r} +  \frac{B' p_C}{Br}  +   \frac{B' C}{Br} \\
&\qquad  + (N-3) \frac{2C}{r^2} +  \delta' p_C + \frac{C}{r} \delta'    \biggr|  \nonumber \\ 
&\quad  + \biggl|  \frac{2}{B r^2} \Lambda_{N - 4} + \frac{4  C}{(N-2)a^2 r^2} p_f^2 \sin^2f \nonumber\\ &\qquad +  \frac{2 C }{a^2 Br^4}  \sin^4 f - \frac{4\Lambda C}{(N-2) B} \biggr|  , \nonumber
\\
\left|  J_f\right|_U & \le   \Biggl|\delta'  +  \frac{B' }{B}  +  (N-4)  \frac{ C' }{2C}  \nonumber\\
&\quad -  \frac{1}{r}\left((N-2) \frac{F_\pi^2}{8} + (N-4) \frac{\sin^2 f }{a^2 r^2} + \frac{ p_f \sin 2f }{2a^2 r^2} \right)\nonumber\\
&\qquad\times \left(\frac{F_\pi^2}{8} + \frac{\sin^2{f}}{a^2 r^2}\right)^{-1}  \Biggr| |p_f|   \nonumber\\
\nonumber  &\quad + \left|  \left(\frac{F_\pi^2}{4} + \frac{\sin^2{f}}{a^2 r^2}\right) \left(\frac{F_\pi^2}{8} + \frac{\sin^2{f}}{a^2 r^2}\right)^{-1}  \right| \left| \frac{\sin{2f}}{2 B r^2} \right|   .
\end{align}
Since both  $C(r)$, and $f(r)$ belong  at least to a class of $C^2$-real functions, then their values are bounded on any closed interval  $I$. Thus, $\left|  \mathcal{J}( {\bf u}, r) \right|_U$ is bounded on $U$.

Moreover, for $\bf{u}, \tilde{\bf{u}} \in U$, we also have
\begin{align}
&\left|  \mathcal{J}( {\bf u}, r) -  \mathcal{J}( \tilde{ {\bf u}}, r) \right|_U \le   \frac{1}{2}  (N - 6)\left| \frac{ p_C^2 }{C} - \frac{ \tilde{p}_C^2 }{\tilde{C}}\right| \\
&\qquad + (N-3) \frac{2}{r}  \left|p_C -\tilde{p}_C \right| +   \frac{1}{r}  \left|  \frac{B' p_C}{B}  -  \frac{\tilde{B}' \tilde{p}_C}{\tilde{B}}  \right|   \nonumber \\ 
&\qquad  + \frac{1}{r}  \left|  \frac{B' C}{B}   -  \frac{\tilde{B}' \tilde{C}}{\tilde{B}}  \right|    +  (N-3) \frac{2}{r^2}  \left| C -\tilde{C} \right|  + \left| \delta' p_C  - \tilde{\delta}' \tilde{p}_C \right|   \nonumber \\ 
&\qquad + \frac{1}{r}  \left| C \delta'   - \tilde{C} \tilde{\delta}'  \right|   + \frac{4  }{(N-2)a^2 r^2} \left| C p_f^2 \sin^2f - \tilde{C}  \tilde{p}_f^2 \sin^2 \tilde{f}   \right|   \nonumber \\ 
&\qquad +  \left| \frac{2 \Lambda_{N-4}}{B \tilde{B} r^2}  \right|  \left| m - \tilde{m}   \right|  +  \frac{2  }{a^2 r^4} \left| \frac{C}{B} \sin^4f - \frac{\tilde{C}}{\tilde{B}}   \sin^4 \tilde{f}   \right|   \nonumber\\
&\qquad  + \frac{4\Lambda }{(N-2) } \left| \frac{C}{B} - \frac{\tilde{C}}{\tilde{B}}   \right| 
 + \left| \delta' +  (N-4)  \frac{ p_C }{2C}   \right| |p_f - \tilde{p_f}| \nonumber\\ 
&\qquad+ \left| \delta'  - \tilde{\delta'}  +  \frac{1}{2} (N-4)  \left(  \frac{p_C}{C}  -\frac{\tilde{p}_C}{\tilde{C} }   \right) \right| |\tilde{p_f}| \nonumber\\
 &\qquad +  \left| \frac{ B' }{B} \right| |p_f - \tilde{p}_f | +  \left| \frac{ B' }{B} -  \frac{ \tilde{B}' }{\tilde{B}}\right| | \tilde{p}_f | \nonumber\\
&\qquad +  \Bigg|  \frac{\sin{2f}}{B r^2} \left(\frac{F_\pi^2}{4} + \frac{\sin^2{f}}{a^2 r^2}\right) \left(\frac{F_\pi^2}{4} + \frac{2 \sin^2{f}}{a^2 r^2}\right)^{-1}  \nonumber\\
&\qquad -  \frac{\sin{2 \tilde{f} } }{ \tilde{B} r^2}  \left(\frac{F_\pi^2}{4} + \frac{\sin^2 \tilde{f}  }{a^2 r^2}\right)  \left(\frac{F_\pi^2}{4} + \frac{2\sin^2 \tilde{f}  }{a^2 r^2}\right)^{-1} \Bigg|   \nonumber\\
&\qquad + \Bigg| \frac{2 p_f}{r} \left( (N-2) \frac{F_\pi^2}{8} + (N-4) \frac{\sin^2 f }{a^2 r^2} + \frac{ p_f \sin 2f }{2a^2 r^2} \right) \left(\frac{F_\pi^2}{4} + \frac{2\sin^2{f}}{a^2 r^2}\right)^{-1} \nonumber\\
&\qquad   -   \frac{2 \tilde{p}_f}{r} \left( (N-2) \frac{F_\pi^2}{8} + (N-4) \frac{\sin^2 \tilde{f}  }{a^2 r^2} + \frac{ \tilde{p}_f \sin 2 \tilde{f}  }{2a^2 r^2} \right) \left(\frac{F_\pi^2}{4} + \frac{2\sin^2 \tilde{f}  }{a^2 r^2}\right)^{-1} \Bigg|    , \nonumber
\end{align}
where we have simplified the notation $\tilde{\delta} \equiv \delta( \tilde{ {\bf u}}, r)$ and $\tilde{B} \equiv B( \tilde{ {\bf u}}, r)$. After some computations using the result in Lemma \ref{lemmalocalLipshitzconstraint} and the local property \eqref{anyFlocal} on $U$, then we obtain
\begin{equation}
\left|  \mathcal{J}({\bf u}, r) -  \mathcal{J}( \tilde{{\bf u}}, r) \right|_U  \le C_{ \mathcal{J}}(|\bf{u}|, |\tilde{\bf{u}}|) | \bf{u} - \tilde{\bf{u}}|  ,  \label{localLipshitzcon}
\end{equation}
proving that $ \mathcal{J}$ is locally Lipshitz with respect to $\bf{u}$.
\end{proof}
%%%%%
%%%%%%

Next, we  write down  \eqref{ReducedEinsteineqSkyrmionEq1} in the integral form 
\begin{equation}
{\bf{u} }(r) = {\bf{u} }(r_H) + \int_{r_H}^{r}\:\mathcal{J}\left( {\bf{u} }(s), s \right)\:ds   . \label{IntegralEquation}
\end{equation}
Then, we define a Banach space
\begin{equation}
X \equiv \{ {\bf{u} } \in C(I,\lR^2) : \: {\bf{u} }(r_H) = {\bf{u} }_{0}, \: \sup_{r\in I}| {\bf{u} }(r)|\leq L_0 \}   ,
\end{equation}
equipped with the norm
\begin{equation}
|{\bf{u} }|_{X} = \sup_{r\in I}\:|\mathbf{u}(r)|  ,
\end{equation}
where $L_0$ is a positive constant. Finally, we introduce an operator $\mathcal{K}$ 
\begin{equation}
\mathcal{K}(\mathbf{u}(r)) = \mathbf{u}_{0} + \int_{r_H}^{r}\:ds \mathcal{J}\left(s,\mathbf{u}(s)\right). \label{OpKdefinition}
\end{equation}
Since $\mathcal{J}$ is locally Lipshitz function with respect to $\mathbf{u}$, then we have the following lemma \cite{Akbar:2015jya}: 
\begin{lemma}
\label{unigueness}
Suppose that we have an operator $\mathcal{K}$  defined in  (\ref{OpKdefinition}) and  there exists a positive constant $\varepsilon$ such that $\mathcal{K}$ is a mapping from $X$ to itself and $\mathcal{K}$ is a contraction mapping on $I = [r_H,r_H + \varepsilon ]$ with
\begin{equation}
\varepsilon   \leq \min\left(\frac{1}{C_{L_0}},\frac{1}{C_{L_0} L_0 + \|\mathcal{J}(r_H)\|}\right)  .
\end{equation}
Then, the operator $\mathcal{K}$ is a contraction mapping on $X$.
\end{lemma}
%%%%
\noindent The existence of a unique fixed point of (\ref{OpKdefinition}) prove that integral equation (\ref{IntegralEquation}) hence differential equation (\ref{ReducedEinsteineqSkyrmionEq1}) has unique local solution. It is worth mentioning that in $N=4$ case we have $(f, p_f)$ as the dynamical variables, whereas  $(\delta, m)$ are still the constraints. By employing  the same logic as above, we could show that they satisfy Lemma \ref{localLipshitz} and Lemma \ref{unigueness}.

Finally, we can construct maximal solution as follows. Suppose ${\bf u}(r)$ is defined on the interval $[r_H, r_m)$ where $r_m$ is a positive constant. Then, by repeating the above arguments of the local existence  with the initial condition ${\bf u}(r-r_0)$ for some $r_H < r_0 < r$ and using  the uniqueness condition to glue the solutions, we obtain the maximal solution. Clearly, if $r_m \to +\infty$, then we have a global solution.

%%%%%%%%%%%%%%%
%%%%%%%%%%%%%%%
\subsection{Global existence}
\label{subsec:GlobExis}

In the final part of this subsection it is necessary to show  that such a regular global solution of \eqref{ReducedEinsteineqSkyrmionEq1} on  $I_{\infty}\equiv [r_H, +\infty)$ does exist satisfying the expansions \eqref{Expandasymp}, \eqref{CrExpandasymp}, and \eqref{fasympcon}. These  ensure that \eqref{IntegralEquation} is finite on $I_{\infty}$ that can be seen as follows. 

Let us first define two intervals, namely, $I_L \equiv [r_H, L]$ and $I_A \equiv (L, +\infty)$ for finite and large $L > r_H$. On $I_A$, all the functions $\delta(r)$, $m(r)$, $C(r)$, and $f(r)$ can be expanded as in  \eqref{Expandasymp}, \eqref{CrExpandasymp}, and \eqref{fasympcon}. Now, we write down  \eqref{IntegralEquation}  as
\begin{equation}
{\bf{u} }(L) = {\bf{u} }(r_H) + \int_{r_H}^{L}\:\mathcal{J}\left( {\bf{u} }(s), s \right)  ds  +  \int_{L}^{+\infty}\:\mathcal{J}\left( {\bf{u} }(s), s \right)  ds . \label{IntegralEquation1}
\end{equation}
One can  show that the third term in the right hand side in  \eqref{IntegralEquation1} converges to zero. Thus, \eqref{IntegralEquation} is finite and globally well-defined since the functions $C(r)$ and  $f(r)$ are at least $C^2$-functions.

Next, we want to show whether  \eqref{IntegralEquation} could have finite energy. This can be done by showing the finiteness of  the energy functional \eqref{Estatic}. Then, using the expansions \eqref{Expandasymp}, \eqref{CrExpandasymp}, and \eqref{fasympcon}, we obtain an inequality
\begin{align} 
\label{Estaticineq} 
E & \le  4\pi A({\mathcal N}^{N-4})  \sup_{r\in I_L} \Bigg|\int_{r_H}^L dr   r^{N-2}  C^{\frac{(N-4)}{2}} e^{\delta} \left[\frac{F^2_{\pi }}{8} \left( B \left(f'\right)^2 + \frac{2\sin^2 f}{r^2} \right) \right.  \\
&\quad  \left.  +  \frac{\sin^2 f}{a^2 r^2} \left( B (f')^2 +\frac{\mathrm{sin}^2 f\ }{ r^2}\right)\right] \Bigg| \nonumber \\
&\quad +   C_0^{\frac{(N-4)}{2}} \Bigg| \int_{L}^{+\infty} dr   r^{N-2} \left[\frac{F^2_{\pi }}{8} \left( - \frac{2 \Lambda n_1^2 \tilde{f}^2_1}{(N-1)(N-2) r^{2n_1}}  + \frac{2\sin^2 f_0}{r^2} \right) \right. \nonumber \\
&\quad \left.  - \frac{2 \Lambda n_1^2 \tilde{f}^2_1 \sin^2 f_0}{(N-1)(N-2) a^2 r^{2n_1+2}}  +\frac{\mathrm{sin}^4 f_0 }{ a^2 r^4}\right] \Bigg|. \nonumber
\end{align} 
 The first term in the right hand side of (\ref{Estaticineq}) is  finite due to  the boundedness of all $C^2$-functions, namely, $B(r)$, $C(r)$, $\delta(r)$ and $f(r)$ on the closed interval $I_L$.  In order to control the second term in (\ref{Estaticineq})  on the open interval $I_A$, after using \eqref{Expandasymp} and \eqref{CrExpandasymp}  one has to set  $f_0 $ and the order $n_1$ in  \eqref{fasympcon} to be  
\begin{equation} 
f_0 = 0  , \quad n_1 > \frac{1}{2} (N-1)  , \label{fasympcon0}
\end{equation}
 for finite $\tilde{f}_1$.
 
In the end, since we have the Einstein's field equation, to obtain a consistent result one has to check the finiteness of
\begin{align} 
\label{G00} 
&  \int_{r_H}^{+\infty} dr   r^{N-2}  C^{\frac{(N-4)}{2}} B^{-1} \left[ R_{00} - \frac{1}{2} g_{00} R + \Lambda g_{00}\right] \\
  &\quad = \int_{r_H}^{+\infty} dr   r^{N-2}  C^{\frac{(N-4)}{2}} e^{\delta} \left[  - \frac{B}{r^2}- \frac{B'}{r^3} + \frac{1}{r^2} \left(1 + (N-4)\frac{\Lambda_{N-4}}{2C}\right)  \right. \nonumber\\
  &\qquad -\Lambda + (N-4)  \left(\frac{B \left( r^2 C\right)''}{2r^2 C} - \frac{B \left( r^2 C\right)'}{r^3 C} - \frac{B' \left( r^2 C\right)'}{4r^2 C}  \right)\nonumber\\
 &\qquad\left. - (N-4)(N-7)  \frac{B {\left( r^2 C\right)'}^2 }{8r^4 C^2} \right]. \nonumber
\end{align}
By using the conditions \eqref{Expandasymp}, \eqref{CrExpandasymp}, and \eqref{fasympcon}, and repeating similar steps as in \eqref{Estaticineq}, we conclude that the integral \eqref{G00} is finite if $\Lambda = 0$. This means that the black holes with finite energy exist only if the black hole spacetimes are asymptotically Ricci flat.

\indent Thus, we could state the following theorem. 
\begin{theorem}
\label{MainresultExistence}
 Let ${\bf{u} }(r)$ be a solution of \eqref{ReducedEinsteineqSkyrmionEq1} with the initial value $\mathbf{u}_H$, $\Lambda \le 0$. Then, there exists a global solution satisfying \eqref{Expandasymp}, \eqref{CrExpandasymp}, and   \eqref{fasympcon} which  interpolates  two boundaries, namely the horizon and the asymptotic regions. Moreover, the solution has finite energy if it satisfies \eqref{fasympcon0} in the asymptotic region with $\Lambda = 0$.
\end{theorem}

\section{Linear stability analysis: a remark}
\label{sec:StableSol}

Finally, it is of interest to check the linear stability of the models using similar method as  \cite{ShiikiPRD2005}.  The main idea is to take the metric functions $\delta(r), m(r), f(r)$ to be also time dependent in advance, namely, $\delta_t(r, t),\linebreak m_t(r, t), f_t(r, t)$ and then expand them around $\delta(r), m(r), f(r)$. The detailed computation of this method can be found in Appendix \ref{sec:Strum-LiouvilleEq}, where we obtain that the linear stability is determined by the perturbation dynamics of $f(r)$.

Writing $f_t(r,t) = f (r) + \varepsilon \frac{\psi(r)e^{i\omega t}}{\sqrt{\alpha}}$ (see \eqref{linearizationmetf} and \eqref{subtitusipertskalar}), where $\varepsilon$ is a small parameter, will yield an eigenvalue problem, i.e.\ Sturm-Liouville equation, in terms of $\psi(r)$ and $\omega$ (see \eqref{LinSta}). %Let us first consider \eqref{LinSta}. 
We say the solution of \eqref{ReducedEinsteineq}, \eqref{ReducedEinsteineq2}, \eqref{ReducedEinsteineq3} and \eqref{SkyrmionEq} to be linearly stable if all spectra of the eigenvalue problem satisfy $\omega^2 > 0$. Using the Sturm-Liouville argument, this is similar to requiring the existence of a positive definite solution (except at the boundaries), {i.e.}\ $\psi(r)>0$, with $\omega^2 > 0$. On the other hand, to show instability, it is sufficient to obtain that there is an eigenvalue with $\omega^2<0$. 

However, it is difficult to solve the eigenvalue problem exactly. Our approach in this paper, i.e., to consider solution behavior near the boundaries, namely, the horizon and the asymptotic region, cannot be used to establish their stability. Therefore, in the following we will only make a remark on the linear eigenvalue problem and some of their spectrum. 

At the horizon $\psi (r_H) = 0$, while in the asymptotic region $\psi(r \to +\infty) \to 0$. Then, using \eqref{Expandasymp}, \eqref{CrExpandasymp}, \eqref{fasympcon},  and \eqref{fasympcon0} in the limit of $r \to +\infty$, the eigenvalue problem %\eqref{LinSta}   
can be simplified to
\begin{align} 
  \label{LinStainfty}
\omega^2 &= -\frac{4 \Lambda^2 r^2}{(N-1)^2 (N-2)^2 \psi} \left(r^2\psi'\right)'  + \frac{N \Lambda^2 r^2}{(N-1)^2 (N-2)}\\
&\quad - \frac{4 \Lambda}{(N-1) (N-2)}  + \frac{2(N-4) n \tilde{C}_1 \Lambda^2}{(N-1)^2 (N-2)^2 C_0} r^{-n+3} \frac{\psi'}{\psi}\nonumber\\
&\quad  -  \frac{2(N-4) n \tilde{C}_1 \Lambda^2}{(N-1)^2 (N-2) C_0} r^{-n+2}  .\nonumber
\end{align}
Suppose there exists a constant $\lambda_0 > 0$ such that we could have, for example, 
\begin{equation} \label{LinStainftypsi}
 -\frac{4 \Lambda^2 r^2}{(N-1)^2 (N-2)^2 \psi} \left(r^2\psi'\right)'  + \frac{N \Lambda^2 r^2}{(N-1)^2 (N-2)} =  \lambda_0   ,  
\end{equation}
whose solution is 
\begin{equation} 
\psi (r) = r^{-1/2} \psi_0   J  \left(-1/2 \sqrt{N(N-2)+1}  ;  \lambda^{1/2}_0/2|\Lambda| (N-1)(N-2) r^{-1/2}\right)  , \label{asympansatzpsiFrobenius}
\end{equation}
where $\psi_0 \in \lR$ and $J( \cdot  ;  \cdot )$ is the Bessel function of the first kind. So, \eqref{LinStainfty} simplifies to
\begin{align} 
 \label{omegavalue}
\omega^2 &= \lambda_0 - \frac{4 \Lambda}{(N-1) (N-2)} \\
\notag &\quad  - \frac{(N-4) n \tilde{C}_1 \Lambda^2}{ 2(N-1)^2 (N-2)^2 C_0}\\
\notag &\qquad\times \left( (N-2)C_0 +1 + \sqrt{N(N-2)+1}\right)  r^{-n+2}  . 
\end{align}
for $n \ge 2$ which is strictly positive for  $\tilde{C}_1 < 0$ and converges to a constant as $r \to +\infty$.

It is worth pointing out that we could have a more general condition, that is, there exists a constant $\lambda_0 \in \lR$ such that
\begin{multline} 
 -\frac{4 \Lambda^2 r^2}{(N-1)^2 (N-2)^2 \psi} \left(r^2\psi'\right)'  + \frac{N \Lambda^2 r^2}{(N-1)^2 (N-2)} \\
  + \frac{2(N-4) n \tilde{C}_1 \Lambda^2}{(N-1)^2 (N-2)^2 C_0} r^{-n+3} \frac{\psi'}{\psi} 
  -  \frac{2(N-4) n \tilde{C}_1 \Lambda^2}{(N-1)^2 (N-2) C_0} r^{-n+2} =  \lambda_0   .  \label{LinStainftypsiumum}
\end{multline}
We could simply solve \eqref{LinStainftypsiumum}, for example, using Frobenius method by taking $n$ to be positive integer ($n \ge 1$). 

Further analysis for the stability of the black holes is addressed for future work.

\section{Comments on the model with $\Lambda > 0$}
\label{sec:Lambdapositif}

This section is devoted to discuss a model with positive cosmological constant $\Lambda$ which modifies some arguments in the preceding sections.  This is so because the physical black holes are defined only on the interval $I_C \equiv [r_H, r_C)$ where $r_C$  is the radius of the outer horizon called cosmological horizon. The discussion around the event horizon with $r = r_H$ in subsection~\ref{subsec:NearHor} still remains, while in subsection \ref{subsec:AsympReg} some arguments  should be slightly modified. At the end, we comment the results in Theorem \ref{MainresultExistence} and Theorem \ref{MainresultLinearStab}.

First, we discuss the solution at the horizon  $r = r_H$. In the case of  $N=4$ and $N=5$, we could have  ${\mathcal T}^2$ either $\lR^2$ or  $AdS_2$ if it satisfies \eqref{SyaratN=4kedua} and \eqref{SyaratN=5kedua}, respectively.
%  
%
% However, this model is unstable in the context of the linear stability discussed in \cite{ShiikiPRD2005}.
%
In $N \ge 6$ case, we get the same formula as  \eqref{SyaratfisisLambdanegatif} with additional condition
%
%In the case of $N \ge 6$, for $\Lambda > 0$ and 
%
\begin{eqnarray}
(N-4)\frac{\Lambda_{N-4}}{C_H} \ge 2\Lambda    , 
\label{SyaratfisisLambdapositif}
\end{eqnarray}
with ${\mathcal T}^2$ either $\lR^2$ or  $AdS_2$. Thus, in this case at hand, we only have primary branches.

Since at $r = r_C$ we also have $B(r_C ) = 0$,  the behavior of $m(r)$ and $\delta(r) $ around the region should be
\begin{equation}
m(r) \rightarrow M > 0  ,\qquad \delta(r) \rightarrow \delta_C \quad \mathrm{as} \quad r \rightarrow r_C  , \label{cosmoradcon}
\end{equation}
with $\delta_C$ a real constant such that the spacetime ${\mathcal M}^N$ becomes Einstein with $\Lambda > 0$. Furthermore,  the functions $C(r) $, and $f(r)$ should also be
\begin{equation}
C(r) \rightarrow C_C  > 0  ,\qquad f(r) \rightarrow f_C \quad \mathrm{as} \quad r \rightarrow r_C  , \label{cosmoradcon1}
\end{equation}
where $f_C $ is a real constant. The Komar integral \eqref{komarint} implies that the mass of the black hole equals $M$ in this case. It is worth mentioning that the spacetime ${\mathcal M}^N$ does not have to converge to Einstein as $r \rightarrow r_C$. If this is the case, then the functions $C(r)$ and $f(r)$ do not converge to constants.

Next, we want to remark about Theorem \ref{MainresultExistence} for $\Lambda > 0$. To establish local existence and uniqueness of $\mathbf{u}(r)$ we obtain the same result since the analysis is locally on $[r_H, r_H + \epsilon ]$ for $\epsilon > 0$. To prove global existence of $\mathbf{u}(r)$, we first construct maximal solution using the local  existence results and then, glue them using the uniqueness condition such that $\mathbf{u}(r)$ can be defined on the interval $[r_H, r_m )$ for $r_H < r_m < r_C$. The finiteness of the energy functional \eqref{Estatic} shows  that  it is not necessary to have the condition \eqref{fasympcon0} since the energy functional \eqref{Estatic} is indeed finite on the interval $I_C$. In other words, we could still have $f_C $ to be a real constant.

Finally, we want to point out that our asymptotic method in Section \ref{sec:Strum-LiouvilleEq} cannot be applied to the case of $\Lambda > 0$ since  $B(r_C ) = 0$.  This implies that  \eqref{LinSta} becomes trivial and we cannot establish a similar result for $\Lambda > 0$.

\section{Conclusion and outlook}
\label{sec:conclusion}

We have constructed a family of static black holes in the higher dimensional Einstein-Skyrme model with the cosmological constant $\Lambda$ where the most part of the discussion in this paper is for $\Lambda \le 0$ case.  The geometry of the black holes  is conformal to $ \mathcal{M}^4 \times \mathcal{N}^{N-4}$ endowed with the metric  \eqref{AnsatzMetric} where $\mathcal{M}^4$ and $\mathcal{N}^{N-4}$ are  the four dimensional spacetime  and the compact  $(N-4)$-dimensional Einstein submanifold with the cosmological constant $\Lambda_{N-4} \ge 0$ , respectively. The Skyrme field is an $SU(2)$ valued field defined locally on the hypersurface  $\mathcal{S}^3 \subseteq \mathcal{M}^4 $ whose form is particularly given by  \eqref{UAnsatz}. This implies that the black holes in general have local topological numbers.

The behavior of the metric functions $\delta(r), ~ m(r)$,  $C(r)$, and the profile function $f(r)$ near the horizon converge to the positive constants $\delta_H , ~ m(r_H) $,  $C_H $, and $f_H$, respectively. The analysis on the Ricci scalar shows that these constants should satisfy the inequality \eqref{geomhorizoncon} in order to have physical black holes. Furthermore, we have so called primary and secondary branches which are simply determined by the value of the cosmological constant $\Lambda$. In $N=4,5$ dimensions, we could have both the  primary and secondary branches, while in  $N \ge 6$ dimensions, there exists only primary branches.

In the asymptotic region, the functions $\delta(r), ~ m(r)$,  $C(r)$, and $f(r)$ should satisfy \eqref{Expandasymp}, \eqref{CrExpandasymp}, and \eqref{fasympcon} which follows that the geometry becomes Einstein with the cosmological constant $2\Lambda/(N-2)$. We also show using the Komar integral that the black hole mass is constant. 

Then, we  established local existence and uniqueness of black hole solutions of \eqref{ReducedEinsteineq}, \eqref{ReducedEinsteineq2}, \eqref{ReducedEinsteineq3}, and  \eqref{SkyrmionEq} using Picard's iteration and contraction mapping properties. % to prove the local existence and the uniqueness. 
We further obtain global existence using the maximal solution technique and showed using  the expansions \eqref{Expandasymp}, \eqref{CrExpandasymp}, and \eqref{fasympcon} that this solution is finite and globally regular. Finally, using the 00-component of the Einstein field equation we showed that the black holes have finite energy if it satisfies \eqref{fasympcon0} and the black hole geometries are asymptotically Ricci flat as stated in Theorem \ref{MainresultExistence}.	

While in this work we could make a remark on the corresponding linear eigenvalue problem, i.e., Sturm-Liouville equation, and some of its spectrum, stability of the black holes is still an open problem. Our asymptotic analysis cannot establish stability as one needs information of the eigenvalue problem on the whole domain. One possibility is to address it by solving the governing equation and the Sturm-Liouville problem numerically, which is planned for future work.

\section*{Acknowledgments}

 BEG would like to thank Ardian Nata Atmaja for useful discussion and Eugen Radu for private communication. The work in this paper is supported by Program WCP  Kemenristekdikti 2018. BEG acknowledges the Abdus Salam ICTP for Associateship 2019 and the warmest hospitality where the final part of this paper was done.

\appendix

\section{Sturm-Liouville equation}
\label{sec:Strum-LiouvilleEq}

This section is devoted to support our  discussion on the linear stability of Einstein-Skyrme model in even dimension. First of all, let us  consider the metric	(see \eqref{AnsatzMetric})
\begin{align} \label{metrik1}
ds^2 &= - e^{2\delta_t(r, t)}B_t(r,t) dt^2 + B_t(r, t)^{-1}dr^2 \\
\notag &\quad + r^2 \left(d\theta^2 + \sin^2 \theta \;d\varphi^2\right) + r^2 C_t(r, t) \hat{g}_{ij}(x^i) \;dx^i dx^j  ,
\end{align}
and  the ansatz for the skyrmion field $U$ (see \eqref{UAnsatz})
\begin{equation}
\label{UAnsatztime}
U = e^{i \vec{\sigma}\cdot\hat{n}f_t(r, t)} = \cos f_t(r, t) + \mathrm{i}\; \hat{n}\cdot\vec{\sigma} \sin f_t(r, t)  .
\end{equation}
Then, using \eqref{metrik1} and  \eqref{UAnsatztime} we can write down the components of Einstein field equation related to \eqref{ESAction}  in the following: 
%
%\textbf{Komponen 00} : $G_{00} + g_{00}\Lambda = T_{00}$
\begin{align} \label{EinsEq00}
&-\frac{B_t}{r^2} - \frac{B'_t}{r} - \Biggl(e^{-2\delta_t}\frac{\dot{B}_t\dot{C}_t}{4B^2_t C_t} + \frac{3B_t}{r^2} + \frac{3B_t C'_t}{rC_t}\\
\notag &  + \frac{B_t C''_t}{2 C_t} + \frac{B'_t C'_t}{4C_t} + \frac{B'_t}{2r}\Biggr)(N-4)
 + e^{-2\delta_t}\frac{\dot{C}_t^2}{8B_t C_t^2}(N-4)(N-5) \\
\notag &- \left(\frac{B_t}{2r^2} + \frac{B_t {C'_t}^2}{8 C_t^2} + \frac{B_t C'_t}{2r C_t}\right)(N-4)(N-7) + \frac{1}{r^2} + \frac{\hat{R}}{2r^2C_t} - \Lambda \\
\notag &= \frac{F_{\pi}^2}{8} \left(e^{-2\delta_t}\frac{\dot{f}_t^2}{B_t} + B_t {f'_t}^2 + \frac{2\sin^2f_t}{r^2}\right)\\
\notag &\quad + \frac{\sin^2f_t}{a^2r^2} \left(e^{-2\delta_t}\frac{\dot{f}_t^2}{B_t} + B_t {f'_t}^2 + \frac{\sin^2f_t}{2r^2}\right)  , 
\end{align}
%
%
%\textbf{Komponen 11} : $G_{11} + g_{11}\Lambda = T_{11}$
\begin{align} \label{EinsEq11}
&\frac{1}{r} \left( 2B_t\delta'_t + B'_t\right)-e^{-2\delta_t} \frac{\dot{C}_t^2}{8 B_t C_t^2}(N-4)(N-7) \\
\notag &+ \frac{B_t}{8}\left(\frac{2}{r} + \frac{C'_t}{C_t}\right)^2(N-4)(N-5) - \frac{1}{r^2} - \frac{\hat{R}}{2r^2C_t} + \Lambda \\
\notag & + \Biggl( -e^{-2\delta_t} \frac{\ddot{C}_t}{2B_t C_t} + e^{-2\delta_t} \frac{ \dot{C}_t}{4 B^2_t C_t} \left(2B_t \dot{\delta}_t + \dot{B}_t \right)\\
\notag & + \left(\frac{B_t}{r} + \frac{1}{4} \left(2B_t \delta'_t + B'_t \right)\right) \left(\frac{2}{r} + \frac{C'_t}{C_t}\right)\Biggr)(N-4) \\
 \notag &= \frac{F_{\pi}^2}{8} \left(e^{-2\delta_t} \frac{\dot{f}_t^2}{B_t} + B_t {f'}_t^2 - \frac{2\sin^2f_t}{r^2}\right) \\
\notag &\quad + \frac{1}{a^2r^2}\sin^2f_t \left(e^{-2\delta_t}\frac{\dot{f}_t^2}{B_t} + B _t {f'}_t^2 - \frac{\sin^2f_t}{2r^2}\right)  ,
\end{align}
%
%\textbf{Komponen 22} : $G_{22} + g_{22}\Lambda = T_{22}$
\begin{align} \label{EinsEq22}
&\frac{e^{-2\delta_t}}{B_t}\left(\frac{\ddot{B}_t}{2B_t} - \frac{\dot{B}_t^2}{B_t^2} - \frac{\dot{\delta}_t \dot{B}_t}{2B_t}\right) + \frac{1}{2}\left(3B'_t \delta'_t + 2B_t \delta''_t + B''_t + 2B_t{\delta'_t}^2\right)\\
\notag & + \frac{1}{r}\left(B_t \delta'_t + B'_t + \frac{B_t}{2r}(N-4)\right)(N-3) %\\
\end{align}
\begin{align}
\notag &+ \Biggl[\frac{e^{-2\delta_t}}{2 B_t C_t}\left(\dot{\delta}_t \dot{C}_t + \frac{\dot{B}_t}{B_t}\dot{C}_t - \ddot{C}_t\right) + \frac{B_t}{2}\left(\frac{(N-2)}{r}\frac{C'_t}{C_t} + \frac{C''_t}{C_t}\right)\\
\notag & + \frac{1}{2}\left(B_t \delta'_t + B'_t\right)\frac{C'_t}{C_t}\Biggr](N-4) \\ 
\notag &- \frac{1}{8}\left(e^{-2\delta}\frac{\dot{C}_t^2}{B _t C_t^2} - B_t \frac{{C'_t}^2}{C_t^2}\right)(N-4)(N-7) - \frac{\hat{R}}{2r^2C_t} + \Lambda\\
\notag &\quad  = \frac{F_{\pi}^2}{8} \left(e^{2\delta_t}\frac{\dot{f}_t^2}{B_t} - B_t {f'_t}^2\right) + \frac{1}{2a^2}\frac{\sin^4f_t}{r^4}  ,
\end{align}\enlargethispage{2.5em}
%
%\textbf{Komponen 33} : $G_{33} + g_{33}\Lambda = T_{33} \implies G_{22}\sin^2\theta + g_{22}\sin^2\theta\Lambda = T_{22}\sin^2\theta \implies G_{22} + g_{22}\Lambda = T_{22}$\\\\
%
%\textbf{Komponen 01} : $G_{01} + g_{01}\Lambda = T_{01}$
\begin{align}  \label{EinsEq01}
& -\frac{\dot{B_t}}{2B_t r}(N-2) + \frac{1}{4}\Biggl[\left(2\delta'_t + \frac{B'_t}{B_t} - \frac{2}{r}\right)\frac{\dot{C}_t}{C_t}\\
\notag &\quad + \left(\frac{\dot{C}_t}{C_t} - \frac{\dot{B}_t}{B_t}\right)\frac{C'_t}{C_t} - \frac{2\dot{C}'_t}{C_t}\Biggr](N-4)\\
\notag &\qquad = 2\dot{f}_t f'_t \left(\frac{F_{{\pi}}^2}{8} + \frac{1}{a^2}\frac{\sin^2f_t}{r^2}\right),
\end{align}
%
%\textbf{Komponen ij} : $G_{ij} + g_{ij}\Lambda = T_{ij}$
\begin{align} \label{EinsEqij}
& \frac{e^{-2\delta_t}}{B_t}\left(\frac{\ddot{B}_t}{2B_t} - \frac{\dot{B}_t^2}{B_t^2} - \frac{\dot{\delta}\dot{B}_t}{2B_t}\right) + \frac{1}{2}\left(3B'_t \delta'_t + 2B_t \delta''_t + B''_t + 2B_t {\delta'_t}^2\right)\\
\notag & + \frac{1}{r}\left(B_t \delta'_t + B'_t + \frac{B_t}{2r}(N-4)\right)(N-3) \\
\notag & + \Biggl[\frac{e^{-2\delta_t}}{2B_t C_t}\left(\dot{\delta}_t \dot{C}_t + \frac{\dot{B}_t}{B}\dot{C}_t - \ddot{C}_t\right) + \frac{B_t}{2}\left(\frac{(N-2)}{r}\frac{C'_t}{C_t} + \frac{C''_t}{C_t}\right)\\
\notag & + \frac{1}{2}\left(B_t \delta'_t + B'_t \right)\frac{C'_t}{C_t}\Biggr](N-5) \\ 
\notag & - \frac{1}{8}\left(e^{-2\delta}\frac{\dot{C}_t^2}{B_t C_t^2} - B_t \frac{{C'}_t^2}{C_t^2}\right)(N-5)(N-8)\\
\notag & - \frac{1}{r^2}\left(1 + \frac{\hat{R}}{2C_t}\frac{(N-6)}{(N-4)}\right) + \Lambda \\
\notag &\quad = \frac{F_{\pi}^2}{8} \left(e^{-2\delta_t}\frac{\dot{f}_t^2}{B_t} - B_t
{f'_t}^2 - \frac{2\sin^2f_t}{r^2}\right)\\
\notag &\qquad  + \frac{1}{a^2}\frac{\sin^2f_t}{r^2} \left(e^{-2\delta_t}\frac{\dot{f}_t^2}{B_t} - B_t{f'_t}^2 - \frac{\sin^2f_t}{2r^2}\right)   .
\end{align}
The skyrmion field equation of motions has the form
\begin{align}\label{SkyrEOMtime}
& e^{-\delta_t}\alpha \frac{\ddot{f}_t}{ B_t} + \frac{e^{\delta_t}}{2}\left[\left(-e^{-2\delta_t}\frac{\dot{f}_t^2}{ B_t} + B_t {f'}_t^2\right)\alpha_f + \frac{\sin 2f_t}{r^2}\beta + \frac{\sin^2 f_t}{r^2}\beta_f\right] \\
\notag &\qquad = (N-4)\frac{\alpha}{2 C_t} C'_t e^{\delta_t} B_t f'_t + \left(e^{\delta_t} B_t \alpha f'_t\right)'  ,
\end{align}
where we have introduced $\alpha \equiv r^{N-4} \left(\frac{ F_{\pi}^2}{8 r^2} + \frac{\sin^2 f_t}{a^2}\right)$,  $\beta  \equiv r^{N-4} \left(\frac{ F_{\pi}^2}{4 r^2} + \frac{\sin^2 f_t}{2a^2}\right)$, $\alpha_f \equiv \frac{\partial\alpha}{\partial f}$, and $\beta_f \equiv \frac{\partial\beta}{\partial f}$.

The next step is to employ the linearization of the metric functions
\begin{align}
\begin{split}
m_t(r,t) &= m (r) + \varepsilon m_{1t} (r,t)  , \\
f_t(r,t) &= f (r) + \varepsilon f_{1t} (r,t)   ,\\
\delta_t(r,t) &= \delta (r) + \varepsilon \delta_{1t} (r,t)   , \\
C_t(r, t) &= C(r)   ,
\end{split}
\label{linearizationmetf}
\end{align}
where $0 < \varepsilon \ll 1$ and $m (r), f (r),  \delta (r), C (r)$ are the solutions of \eqref{ReducedEinsteineq}-\eqref{ReducedEinsteineq3} and \eqref{SkyrmionEq} such that from \eqref{EinsEq00}- \eqref{SkyrEOMtime} we have a set of linearized equations
\begin{align}\label{linearEinsEq00}
&-\frac{m_{1t}'}{r} \left(1 + \frac{1}{4}\frac{(rC' + 2C)}{C}(N-4)\right) \\
\notag & - \frac{m_{1t}}{r^2} \Biggl[1 + \frac{1}{2} \frac{\left(2C + 4C' r + C'' r^2\right)}{C}(N-4) \\
\notag & + \frac{(C'r + 2C)}{C}(N-4) \left(1 + \frac{1}{8} \frac{C'r + 2C}{C}(N-7)\right)\Biggr] \\
\notag &\quad = \frac{F_{\pi}^2}{8} \left(m_{1t}{f'}^2 + 2Bf_{1t}' f' + \frac{2f_{1t}\sin2f}{r^2}\right) \\
\notag &\qquad + \frac{1}{2a^2} \Biggl(\frac{2m_{1t}{f'}^2\sin^2f}{r^2} + \frac{4Bf_{1t}'f'\sin^2f}{r^2} \\
\notag &\qquad\quad + \frac{2B{f'}^2f_{1t}\sin2f}{r^2} + \frac{2f_{1t}\sin2f\sin^2f}{r^4}\Biggr)  ,
\end{align}
\begin{align}\label{linearEinsEq11}
&\frac{m_{1t}'}{r} \left(1 + \frac{1}{4} \frac{(C'r + 2C)}{C}(N-4)\right)\\
\notag & + \frac{B\delta_{1t}'}{r}(N-2) + \frac{1}{2}\frac{B\delta_{1t}'C'}{C}(N-4) \\
\notag &+ \frac{m_{1t}}{r^2} \Biggl[1 + 2\delta'r + \frac{C'r + 2C}{C}(N-4)\\
\notag &\qquad\times  \left(1 + \frac{1}{2}\delta'r + \frac{1}{8}\frac{(C'r + 2C)}{C}(N-5)\right)\Biggr] \\
\notag &\quad = \frac{F_{\pi}^2}{8} \left(m_{1t}{f'}^2 + 2Bf_{1t}'f' - \frac{2f_{1t}\sin2f}{r^2}\right) \\
\notag &\qquad + \frac{1}{2a^2} \Biggl(\frac{2m_{1t}{f'}^2\sin^2f}{r^2} + \frac{4Bf_{1t}' f'\sin^2f}{r^2} \\
\notag &\qquad\quad+ \frac{2B{f'}^2f_{1t}\sin2f}{r^2} - \frac{2f_{1t}\sin2f\sin^2f}{r^4}\Biggr)  ,
\end{align}
\begin{align} \label{linearEinsEq22}
& e^{-2\delta} \frac{m_{1t}}{2B^2} + \frac{1}{2} \Bigl(m_{1t}'' + 2\delta'' m_{1t} + 2\delta_{1t}'' B + 2{\delta'}^2 m_{1t}\\
\notag &\qquad  + 4\delta' \delta_{1t}' B + 3\delta' m_{1t}' + 3\delta_{1t} B'\Bigr) \\
\notag &+ \frac{m_{1t}}{r^2} \Biggl[(N-3) \left(\delta'r + \frac{1}{2}(N-4)\right) \\
\notag &\qquad + \frac{r^2}{2C}(N-4) \left(C'' + \frac{C'}{r}(N-2) + \frac{1}{4}\frac{{C'}^2}{C}(N-7) + \delta'C'\right)\Biggr] \\
\notag &+ \frac{B\delta_{1t}'}{r}(N-3) + \frac{B\delta_{1t}'C'}{2C}(N-4) \\
\notag &+ \frac{m_{1t}'}{r} \left((N-3) + \frac{1}{2}\frac{C'r}{C}(N-4)\right) \\
\notag &\qquad = - \frac{F_{\pi}\delta^2}{8} \left(m_{1t}{f'}^2 + 2B f_{1t}' f'\right) + \frac{1}{a^2} \frac{f_{1t} \sin2f\sin^2f}{r^4}  ,
\end{align}
\begin{align} \label{linearEinsEq01}
&- \frac{\dot{m}_{1t}C'}{4BC}(N-4)
- \frac{1}{2}\frac{\dot{m}_{1t}}{Br}(N-2)\\
\notag &\qquad  = \frac{F_{\pi}^2}{4}\dot{f}_{1t} f' + \frac{2}{a^2}\frac{\dot{f}_{1t} f'\sin^2f}{r^2}   ,
\end{align}
\begin{align} \label{linearEinsEqij}
& e^{-2\delta} \frac{m_{1t}}{2B^2} + \frac{1}{2} \Bigl(m_{1t}'' + 2\delta'' m_{1t} + 2\delta_{1t}'' B + 2{\delta'}^2 m_{1t} \\
\notag &\qquad + 4\delta' \delta_{1t}' B + 3\delta' m_{1t}' + 3\delta_{1t} B'\Bigr) \\
\notag &+ \frac{m_{1t}}{r^2} \Biggl[(N-3) \left(\delta'r + \frac{1}{2}(N-4)\right) \\
\notag &+ \frac{r^2}{2C}(N-5) \left(C'' + \frac{C'}{r}(N-2) + \frac{1}{4}\frac{{C'}^2}{C}(N-8) + \delta'C'\right)\Biggr] \\
\notag &+ \frac{B\delta_{1t}'}{r}(N-3) + \frac{B\delta_{1t}'C'}{2C}(N-5) \\
\notag &+ \frac{m_{1t}'}{r} \left((N-3) + \frac{1}{2}\frac{C'r}{C}(N-5)\right) \\
\notag &\quad = - \frac{F_{\pi}^2}{8} \left(m_{1t}{f'}^2 + 2B f_{1t}' f' + \frac{2f_{1t} \sin2f}{r^2}\right) \\
\notag &\qquad- \frac{1}{2a^2} \Biggl(\frac{2m_{1t}{f'}^2\sin^2f}{r^2} + \frac{4B f_{1t}' f'\sin^2f}{r^2} \\
\notag &\qquad\quad+ \frac{2B{f'}^2f_{1t} \sin2f}{r^2} + \frac{2f_{1t} \sin2f\sin^2f}{r^4}\Biggr)  ,
\end{align}\enlargethispage{2.5em}
\begin{align} \label{linearSkyrEOMtime}
e^{-\delta} \frac{\ddot{f}_{1t}}{B}\alpha  &= \left(e^{\delta} B \alpha f_{1t}'\right)' + \left(e^{\delta}m_{1t}\alpha f'\right)' + \left(e^{\delta}B\alpha_{f_{1t}}f'\right)' + e^{\delta}\delta_{1t}'B\alpha f'\\
\notag &\quad + \frac{(N-4)e^{\delta}}{2C} \left(BC'f_{1t}' + m_{1t} C'f'\right) \alpha\\
\notag &\quad + \frac{(N-4)}{2}\frac{e^{\delta} BC' f'}{C}\alpha_{f_{1t}} - \frac{e^{\delta}}{2} \left(m_{1t}{f'}^2
+ 2Bf'f_{1t}'\right)\alpha_f \\
\notag &\quad - \frac{1}{2}\left(e^{\delta}B{f'}^2\right)\alpha_{ff_{1t}} - \frac{e^{\delta}\sin2f}{2r^2}\left(\beta_{f_{1t}} + f_1\beta_f \right) \\
\notag &\quad - \frac{e^{\delta}\cos2f}{r^2}f_{1t}\beta - \frac{e^{\delta}\sin^2f}{2r^2}\beta_{ff_{1t}}  ,
\end{align}
where we have expanded $\alpha$ and $\beta$ with respect to $f_t = f + \varepsilon f_{1t}$ to obtain 
\begin{align*}
&\alpha = \frac{F_{\pi}^2}{8}r^{N-2} + \frac{1}{a^2}\sin^2f  r^{N-4},\quad \alpha_{f_{1t}} = \frac{f_{1t}\sin2f}{a^2}r^{N-4},\\ 
&\alpha_f = \frac{\sin2f}{a^2}r^{N-4},\quad \alpha_{ff_{1t}} = \frac{2f_{1t} \cos2f}{a^2}r^{N-4},\\
&\beta = \frac{F_{\pi}^2}{4}r^{N-2} + \frac{\sin^2f}{2a^2}r^{N-4},\quad \beta_{f_{1t}} = \frac{f_{1t}\sin2f}{2a^2}r^{N-4},\\
&\beta_f = \frac{\sin2f}{2a^2}r^{N-4},\quad \beta_{ff_{1t}} = \frac{f_{1t}\cos2f}{a^2}r^{N-4}.
\end{align*}

From \eqref{linearEinsEq00} and \eqref{linearEinsEq11} we obtain
\begin{equation} \label{delone}
\delta_{1t}' = \left[\frac{2 Cr}{C'r(N-4) + 2 C(N-2)}\right] \left[\frac{4f'f_{1t}'}{r^{N-2}}\alpha_{f} + \frac{2{f'}^2f_{1t}}{r^{N-2}} \alpha_{ff}\right]   ,
\end{equation}
while \eqref{linearEinsEq01} gives us 
\begin{equation} \label{B1}
m_{1t} = - B \left[\frac{2 C r}{C'r(N-4) + 2 C(N-2)}\right] \left(\frac{4f_{1t} f'}{r^{N-2}}\alpha_{f}\right)  + \tilde{m}_1(r)  .
\end{equation}
Equations \eqref{delone} and \eqref{B1} show that it is sufficient to consider only the dynamics of $f_{1t}$. Thus, using \eqref{delone} and \eqref{B1}, \eqref{linearSkyrEOMtime} can be simplified to 
\begin{equation} \label{stahig}
\frac{\ddot{f}_{1t}}{K_0^{(1)}}\alpha = \left(K_0^{(1)}\alpha f_{1t}'\right)' + K_0^{(2)}\alpha f_{1t}' - \left(U_0^{(1)} + U_0^{(2)}\right)f_{1t}  ,
\end{equation}
with
\begin{align}
K_0^{(1)} &= e^{\delta} B  , \nonumber \\
K_0^{(2)} &= (N-4)\frac{e^{\delta}BC'}{2C}   , \nonumber \\
U_0^{(1)} &= \left(\frac{8e^{\delta}B\alpha^2{f'}^2C}{\left(C'r(N-4) + 2C(N-2)\right)r^{N-3}}\right)' - \left(\frac{e^{\delta}Bf'\sin2f}{a^2}r^{N-4}\right)'    \nonumber \\
&\quad - \frac{8e^{\delta}B \alpha \alpha_f {f'}^3C}{(C'r(N-4) + 2C(N-2))r^{N-3}}  + \frac{e^{\delta}B{f'}^2\cos2f}{a^2}r^{N-4}  \nonumber \\
&\quad + \frac{e^{\delta} \sin^2 2f}{2a^2}r^{N-6} + \frac{e^{\delta}\cos2f}{r^2}\beta + \frac{e^{\delta}\sin^2f\cos2f}{2a^2}r^{N-6}   ,  \nonumber \\
U_0^{(2)} &= (N-4) \left(\frac{4e^{\delta}B\alpha^2{f'}^2C'}{(C'r(N-4) + 2C(N-2))r^{N-3}} - \frac{e^{\delta}Bf'C'\sin2f}{2Ca^2}r^{N-4}\right)  .  \nonumber
\end{align}
One can further simplify \eqref{stahig} by taking 
\begin{equation}
 f_{1t}(r,t) = \frac{\psi(r)e^{i\omega t}}{\sqrt{\alpha}}  ,  \label{subtitusipertskalar}
 \end{equation}
such that \eqref{stahig}  can be cast into
\begin{align} 
\notag -\frac{\omega^2\psi}{K_0^{(1)}} &= \left(K_0^{(1)}\psi'\right)'\sqrt{\alpha} + K_0^{(2)}\psi' \\
\label{LinSta}
 &\quad - \left[\frac{K_0^{(1)'}\alpha'}{2\alpha} + \frac{1}{2\sqrt{\alpha}}\left(\frac{\alpha'}{\sqrt{\alpha}}\right)'K_0^{(1)} + \frac{K_0^{(2)}\alpha'}{2\alpha} + \frac{U_0^{(1)} + U_0^{(2)}}{\alpha}\right]\psi  . 
\end{align}
Defining a new coordinate $r*$
\begin{equation} \label{EqFFin}
\frac{dr*}{dr} = \frac{1}{K_0^{(1)}}  ,
\end{equation}
Eq.\ \eqref{LinSta} could be further simplified into
\begin{equation} \label{Stourm}
\frac{d^2\psi}{dr*^2} + K_0^{(2)}\frac{d\psi}{dr*} + (\omega^2 -  \hat{U}_0 ) \psi = 0 ,
\end{equation}
where
\begin{equation}
\hat{U}_0 = K_0^{(1)}\left[\frac{K_0^{(1)'}\alpha'}{2\alpha} + \frac{1}{2\sqrt{\alpha}}\left(\frac{\alpha'}{\sqrt{\alpha}}\right)'K_0^{(1)} + \frac{K_0^{(2)}\alpha'}{2\alpha} + \frac{U_0^{(1)} + U_0^{(2)}}{\alpha}\right]  ,
\end{equation}
which is a Sturm-Liouville equation.

\address{Indonesian Center for Theoretical and Mathematical Physics (ICTMP)\\
and Theoretical Physics Laboratory\\
Theoretical High Energy Physics Research Group\\
Faculty of Mathematics and Natural Sciences\\
Institut Teknologi Bandung\\
Jl. Ganesha no. 10 Bandung, Indonesia, 40132\\
\email{bobby@fi.itb.ac.id}}

\clearpage
\address{Theoretical Physics Laboratory\\
Theoretical High Energy Physics Research Group\\
Faculty of Mathematics and Natural Sciences\\
Institut Teknologi Bandung\\
Jl. Ganesha no. 10 Bandung, Indonesia, 40132\\
\email{ftakbar@fi.itb.ac.id}}

\address{Theoretical Physics Laboratory\\
Theoretical High Energy Physics Research Group\\
Faculty of Mathematics and Natural Sciences\\
Institut Teknologi Bandung\\
Jl. Ganesha no. 10 Bandung, Indonesia, 40132\\
\email{rizqi.fadli@students.itb.ac.id}}

\address{Theoretical Physics Laboratory\\
Theoretical High Energy Physics Research Group\\
Faculty of Mathematics and Natural Sciences\\
Institut Teknologi Bandung\\
Jl. Ganesha no. 10 Bandung, Indonesia, 40132\\
\email{dedenmidzanulakbar@ymail.com}}

\address{Theoretical Physics Laboratory\\
Theoretical High Energy Physics Research Group\\
Faculty of Mathematics and Natural Sciences\\
Institut Teknologi Bandung\\
Jl. Ganesha no. 10 Bandung, Indonesia, 40132}

\address{Department of Mathematical Sciences, University of Essex\\
Colchester, CO4 3SQ, UK\\
\email{hsusanto@essex.ac.uk}}

\end{document}